\documentclass[12pt]{article}

\usepackage{amsmath,amssymb,amsthm}
\usepackage{mathtools}
\usepackage[margin=1in]{geometry}
\usepackage{enumitem}
\usepackage{booktabs}
\usepackage{microtype}
\usepackage[hidelinks]{hyperref}
\usepackage{xcolor}
\usepackage{tikz}
\usepackage{graphicx}
\usepackage{pdflscape}
\usepackage{caption}
\usepackage{mathrsfs}

\newtheorem{theorem}{Theorem}[section]
\newtheorem{lemma}[theorem]{Lemma}
\newtheorem{remark}[theorem]{Remark}
\newtheorem{proposition}[theorem]{Proposition}
\newtheorem{corollary}[theorem]{Corollary}
\theoremstyle{definition}
\newtheorem{definition}[theorem]{Definition}

\newcommand{\curve}{\mathscr{C}}
\newcommand{\FF}{\mathbb{F}}

\newcommand{\NN}{\mathbb{N}}
\newcommand{\calL}{\mathcal{L}}
\newcommand{\calF}{\mathcal{F}}
\newcommand{\PP}{\mathbb{P}}
\newcommand{\ev}{\operatorname{ev}}
\newcommand{\divv}{\operatorname{div}}
\newcommand{\ord}{\operatorname{ord}}
\newcommand{\supp}{\operatorname{supp}}
\newcommand{\wt}{\operatorname{wt}}
\newcommand{\cosu}{\operatorname{cosupp}}
\newcommand{\cowt}{\operatorname{cowt}}
\newcommand{\van}{Z_{X}}

\title{Generalized Hamming Weights of\\ AJ--Gorenstein One-Point Codes}

\author{
Eliseo Sarmiento Rosales\thanks{Escuela Superior de F\'isica y
Matem\'aticas, Instituto Polit\'ecnico Nacional, M\'exico.}
\and
Jos\'e Alberto Guzm\'an-Vega\footnotemark[1]
\and
Juan Carlos Jim\'enez-Cervantes\footnotemark[1]
}

\date{}

\begin{document}

\maketitle

\begin{abstract}
We study generalized Hamming weights along the one-point code flag of
an AJ--Gorenstein curve.  We organize these weights in a graded array,
the zero diagram, whose entries are generalized coweights: the
largest numbers of evaluation points on which subcodes of prescribed
dimensions vanish simultaneously.  Twisted and Wei duality show that
each row of the zero diagram determines both the generalized Hamming
weights of a lower block of short codes and the missing weights of a
reflected upper block of long codes in the GHW diagram of the complete
flag.  Our main quantitative result is a uniform coverage
theorem.  For an AJ--Gorenstein curve of genus \(g\) and evaluation
length \(n>2g\), the proportion of generalized-weight positions
determined exactly throughout the complete flag
satisfies
\(\operatorname{Cov}_{\mathrm{full}}
\ge \frac{n(n-1)+4g}{n(n+2g-1)}>\frac12\).

Thus more than half of all generalized-weight positions in the
complete flag are determined uniformly.  For the smallest Suzuki
curve, the general and Castle-specific mechanisms together determine
\(2{,}280\) of the \(2{,}912\) positions, giving an exact coverage of
\(78.30\%\).
\end{abstract}

\section{Introduction}\label{sec:introduction}

The generalized Hamming weights (\emph{GHWs}) of a linear code refine its minimum
distance by measuring the smallest support of a subcode of each possible
dimension.  Introduced by Wei~\cite{Wei1991}, they form the \emph{weight
hierarchy} of the code and carry substantially more information than the
first weight alone.  Complete hierarchies, however, are known for
comparatively few families.  For algebraic-geometric codes the problem is
especially delicate: Riemann--Roch theory gives robust bounds, but exact
higher weights depend on the rational-point geometry of linear systems on
the underlying curve.

One-point algebraic-geometric codes come naturally in a flag
\[
C_X(0)\subseteq C_X(Q)\subseteq C_X(2Q)\subseteq\cdots,
\]
where the inclusions are strict precisely at the nongaps of the
Weierstrass semigroup at~$Q$.  Castle curves provide a particularly
structured instance of this construction.  Their Weierstrass semigroup
is symmetric, their number of rational points attains the relevant
Lewittes bound, and the associated one-point flag has a twisted
self-duality; see Munuera, Sep\'ulveda, and
Torres~\cite{MunueraTorres2009}.  These properties make Castle curves a
natural class in which to seek a unified description of generalized
Hamming weights rather than a collection of isolated minimum-distance
computations.  In fact only two of these properties are ever used.  We isolate them
under the name \emph{AJ--Gorenstein} (Definition~\ref{def:aj-gorenstein}):
the evaluation divisor is linearly equivalent to a multiple of~$Q$ (an
Abel--Jacobi condition), and $H(Q)$ is symmetric (a Gorenstein
condition).  Every Castle triple is
AJ--Gorenstein (Proposition~\ref{prop:castle-is-aj-gorenstein}), and all
general results below are proved at this level of generality; Castle
curves reappear in Section~\ref{sec:castle-specialization} as the source
of the extra structure needed for explicit computations.

The main difficulty is that a flag contains many codes, each with its own
weight hierarchy.  Listing these hierarchies separately hides the
relations between them.  We instead organize the relevant information in
a graded array.  If
\[
H(Q)=\{0=\lambda_0<\lambda_1<\lambda_2<\cdots\}
\]
is the sequence of nongaps and
\(C_j\coloneqq C_X(\lambda_jQ)\), then the entry \(M_{j,s}\) is the largest
number of evaluation points on which an \(s\)-dimensional subcode of
\(C_{j}\) vanishes simultaneously.  The row
\[
\mathcal M_j:=(M_{j,1},\ldots,M_{j, k_{j}})
\]
is a profile of a small linear series.  We  call the
resulting graded array the \emph{zero diagram} (Figure~\ref{fig:zero_diagram}, Section~\ref{sec:injective-diagram}); we show in
Section~\ref{sec:postcanonical-region} how it determines the
\emph{GHW diagram} of actual generalized Hamming weights
(Definition~\ref{def:ghw-diagram}).

Our first contribution is an upper--lower block correspondence.  The
\emph{lower block} indexed by~$j$ is the plateau of short codes between
\(\lambda_j\) and \(\lambda_{j+1}\).  Its generalized Hamming weights
are directly \(n-M_{j,s}\).  The reflected plateau near the opposite end
of the flag is the \emph{upper block} indexed by~$j$.  Twisted duality of
the one-point flag, followed by Wei duality, shows that its missing weights
are exactly \(M_{j,s}+1\).  Thus one maximum vanishing number computation fills two
regions of the diagram simultaneously.

We organize the paper around this reflection.  The evaluation map \(\ev_X\) is injective on \(\calL(hQ)\) exactly
for \(h<n\); we call this region the \emph{injective zero diagram}
(Section~\ref{sec:injective-diagram}), and its twisted-duality
extension to the full range \(0\le h\le N\) the \emph{full zero
diagram} (Section~\ref{sec:full-diagram}).  In the postcanonical
range \(2g\le\lambda_j<n\) (Section~\ref{sec:postcanonical-region}),
the upper block of a postcanonical row folds back into another row of
the injective zero diagram itself.

A second structural layer is the zero profile
\(
m_G(e)=\max_{|E|=e}\ell(G-D_E)
\)
(\(D_E\) the sum of the points in \(E\subseteq X\)): linear-code
duality gives a pointwise kernel identity on complementary coordinate
sets, and Riemann--Roch the corresponding identity between \(m_G(e)\)
and \(m_{K+D_X-G}(n-e)\); we retain this divisor form only where its
residual information is essential, notably in the postcanonical
profile.  The injective zero diagram itself is filled using Cartesian
embeddings, gonality bounds, full-support packings, built by
complementation and disjoint union of full supports, i.e.\
rational-point divisors linearly equivalent to a multiple of~\(Q\),
and rank propagation; in the postcanonical region these combine to
give an exact tail, a full/no-full boundary, residual exactness,
full-level staircases, and a uniform density theorem.

The Suzuki curves are used to illustrate the postcanonical
density theorem numerically. Detailed computations, including a
theorem-by-theorem coverage map for the smallest Suzuki curve, are
collected in Appendix~\ref{app:numerical-examples}.  The Suzuki family is a Castle curve family with highly structured one-point codes, but its
Weierstrass semigroup and low-degree embeddings lie outside the simplest
two-generator model. Earlier work determines the Hermitian hierarchy through its special
semigroup structure~\cite{BarberoMunuera2000}, the second weight of
Castle codes with two-generator Weierstrass
semigroups~\cite{OlayaLeonGranados2015}, and, more recently, the
hierarchy of decreasing norm-trace codes through footprint methods~\cite{CampsMoreno2025}.  General order bounds for higher
weights of one-point and multipoint algebraic-geometric codes are
available in \cite{GeilMunueraRuanoTorres2011,BrasAmorosLeeVicoOton2014,Lee2015},
but they need not determine the exact maxima considered here.
For Suzuki codes, the existing literature includes constructions and
distance bounds~\cite{HansenStichtenoth1990,Matthews2004,DuursmaPark2008,
DuursmaKirov2009,FarranEtAl2018,MontanucciTimpanellaZini2018};
the present approach instead targets complete blocks and exact families
of generalized weights.   The GHW diagram itself is
compatible with isometry-dual flags~\cite{BrasAmorosDuursmaHong2020},
and its zero viewpoint is close to earlier geometric approaches to
higher weights~\cite{Munuera1994,YangKumarStichtenoth1994}. Our method thus provides large exact regions of the zero diagram from
the general geometric framework alone.  For the smallest Suzuki curve,
combining these mechanisms determines \(2{,}280\) of the \(2{,}912\)
positions in the GHW diagram, or \(78.30\%\) (see
Table~\ref{tab:suzuki-q8-theorem} for the theorem-by-theorem
breakdown).  More generally, for sequences of Castle curves with
\(g/n\to0\), the proportion of exact cells that the degree-driven
region contributes to the injective zero diagram tends to one.

The main contributions can therefore be summarized as follows.
\begin{enumerate}[label=\textup{(\roman*)}]
\item We introduce the GHW diagram and prove
(Theorem~\ref{thm:block-dictionary}) that each zero row determines
both a lower block and, through twisted duality
(Proposition~\ref{prop:twisted-duality}), its reflected upper block;
the generalized-weight problem for the full flag \(0\le h\le N\) is
thereby reduced to the injective zero diagram \(h<n\).
\item We establish the Postcanonical Profile
(Theorem~\ref{thm:postcanonical-profile}): an exact tail, the
full/no-full canonical boundary, a residual--gonality exactness
criterion, and full-level staircases.  If \(n>2g\), its degree-driven
region alone contains at least \(\binom{n-2g+1}{2}\) exact cells of
the injective zero diagram.
\item Combining this tail with twisted duality, we prove
(Corollary~\ref{cor:full-diagram-coverage}) that every AJ--Gorenstein
triple with \(n>2g\) satisfies
\[
\operatorname{Cov}_{\mathrm{full}}
\ge
\frac{n(n-1)+4g}{n(n+2g-1)}
>
\frac12,
\]
so more than half of all generalized-weight positions of the codes
\(C_X(hQ)\), \(0\le h\le N\), are determined uniformly.
\end{enumerate}

\section{Preliminaries}\label{sec:preliminaries}

\subsection{Linear codes and generalized Hamming weights}
\label{subsec:linear-codes}

\begin{definition}
    Let \(\FF_q\) be the finite field of \(q\) elements, for a vector \(c\in\FF_q^n\), its \emph{support} is the index set of its nonzero entries,
    \[\supp(c):=\{i:c_{i}\neq 0\}.\]
    The size of this set is the \emph{Hamming weight} of \(c\), denoted
    \(\wt(c):=\lvert\supp(c)\rvert.\)
\end{definition}

\begin{definition}
    A \emph{linear code} (or simply a \emph{code}) \(C\) of length \(n\) is a subspace of \(\FF_q^n\).\linebreak Its \emph{support} is the union of the supports of its codewords,
    \[
        \supp(C) :=
        \bigcup_{c\in C}\supp(c)=
        \{i : \exists c\in C\text{ with }c_{i}\neq 0\}.
    \]
    If \(C\) is of dimension \(k\), its \emph{\(r\)-th generalized Hamming weight} (\(1\le r\le k\)) is given by
    \[
        d_r(C) :=
        \min\{\lvert\supp(D)\rvert : D \subseteq C, \ \dim D=r\}.
    \]
    The strictly increasing sequence \(d_{1}(C)<\cdots<d_{k}(C)\) constitutes the \emph{weight hierarchy} of \(C\). 
\end{definition}

\begin{definition}
    Under the dot product in \(\FF_q^n\), the \emph{dual code} of \(C\) is its orthogonal subspace,
    \[C^\perp:=\{x\in\FF_q^n:c\cdot x=0\text{ for all }c\in C\}.\]
    For vectors \(v, c\in\FF_q^n\), their entry-wise product is
    \(v\star c:= (v_{1}c_{1}, \cdots, v_{n}c_{n})\); if \(v\in (\FF_q^*)^n\), scaling \(C^\perp\) by the entry-wise inverse \(v^{-1}\) yields the \emph{scaled dual code},
    \[C^{\perp_{v}} := v^{-1} \star C^\perp.\]
\end{definition}
Since \(C^{\perp_{v}}\) and \(C^\perp\) are monomially equivalent, their generalized Hamming weights coincide.

\begin{theorem}[Wei Duality \cite{Wei1991}]\label{thm:wei}
    Let \(C\) be a code of length \(n\) and dimension \(k\). Then,\linebreak the weight hierarchies of \(C\) and \(C^\perp\) are related by
    \[
        \{d_{r}(C):1\le r\le k\}=[1, n]\setminus\{n+1-d_{s}(C^\perp):1\le s\le n-k\}.
    \]
\end{theorem}
The strictly decreasing sequence \(n+1-d_{1}(C^\perp)>\cdots>n+1-d_{n-k}(C^\perp)\) constitutes the \emph{missing weights} of \(C\).

\subsection{One-point codes and Weierstrass semigroups}
\label{subsec:one-point-codes}
Let \(\curve/\FF_q\) be a smooth, projective, geometrically irreducible curve of genus~\(g\ge1\).
Fix a rational point \(Q\in\curve(\FF_q)\) and an evaluation set \(X=\{P_{1}, \cdots, P_{n}\}\subseteq\curve(\FF_q)\setminus\{Q\}\), which defines the effective divisor \(D_{X}:=\sum_{P\in X}P\) of degree \(n\).
\begin{definition}
    For a divisor \(G\), its \emph{Riemann--Roch space} is the vector space given by
    \[
        \calL(G) := \{f\in\FF_q(\curve)^*: \divv(f)+G\ge0\}\cup\{0\},
    \]
    where we denote \(\ell(G) := \dim_{\FF_q}\calL(G)\). Note that \(\calL(G)=\{0\}\) whenever \(\deg G<0\).
\end{definition}
Throughout, \(K\) denotes a fixed canonical divisor of \(\curve\), so that
\(\deg K=2g-2\) and \(\ell(K)=g\), and we use the Riemann--Roch theorem in
the form
\[
    \ell(G)-\ell(K-G)=\deg G+1-g ;
\]
see \cite[Ch.~1]{Stichtenoth2009} for this and for the semigroup facts
recalled below.
\begin{definition}
    For \(h\in\NN_0\), the \emph{evaluation map} on the space \(\calL(hQ)\) is
    \[\ev_{X} \colon \calL(hQ) \to \FF_q^n, \quad f \mapsto (f(P_1),\ldots,f(P_n)),\]
    which is well defined as \(Q\notin X\).
    The \emph{one-point evaluation code} of order~\(h\) is the image space
    \[C_X(hQ) := \ev_X(\calL(hQ))\subseteq\FF_q^n.\]
\end{definition}
\begin{definition}
    A \emph{numerical semigroup} \(H\) is defined as a co-finite sub-semigroup of \((\NN_0,+)\).
    The elements of $\mathbb{N}_0 \setminus H$ are its \emph{gaps}, and the largest such gap is its \emph{Frobenius number} $F_H$.
    The semigroup \(H\) is \emph{symmetric} if
    \[
        a\in H \quad\Longleftrightarrow\quad F_{H}-a\notin H .
    \]
\end{definition}
\begin{definition}
    The \emph{Weierstrass semigroup} at~\(Q\) is the set of non-negative integers
    \[H(Q):=\{h\in\NN_0:\ell(hQ)>\ell((h-1)Q)\}= \{0=\lambda_0<\lambda_1<\lambda_2<\cdots\};\]
    equivalently, \(h\in H(Q)\) if and only if some function \(f\in\FF_q(\curve)^*\) has pole divisor exactly \(hQ\).
\end{definition}
Throughout this paper, we shall rely on the classical identities
\[
    \ell(hQ)=\lvert H(Q)\cap[0,h]\rvert, \qquad
    \lvert\NN_0\setminus H(Q)\rvert=g,
\]
the first one being valid for every \(h\in\mathbb{Z}\) under the convention that \([0,h]=\emptyset\) for \(h<0\).\linebreak
As for the second, \(H(Q)\) has \(g\) gaps, hence it is symmetric precisely when \(F_{H(Q)}=2g-1\); for Weierstrass semigroups this holds exactly when \(Q\) is \emph{subcanonical}, that is, \(K\sim(2g-2)Q\). In that case \(2g-1\) is a gap while every integer \(\lambda\ge2g\) is not, so
\[
    2g-2<\lambda \quad\Longleftrightarrow\quad 2g\le\lambda \qquad\text{for } \lambda\in H(Q).
\]
In particular, the nongaps from $2g$ onwards are consecutive, so
that for a symmetric \(H(Q)\)
\begin{equation}\label{eq:postcanonical-index}
    \lambda_j=g+j \quad\text{for every } j\ge g,
    \qquad\text{and}\qquad
    \lambda_j\ge2g \iff j\ge g .
\end{equation}
Indeed, for \(h\ge2g-1\) the identity \(\ell(hQ)=\lvert H(Q)\cap[0,h]\rvert\) gives \(\ell(hQ)=h+1-g\), so a nongap \(h=g+j\) is the \((j+1)\)-st element of \(H(Q)\); the excluded index \(j=g-1\) corresponds to the gap \(2g-1\), and \(\lambda_{g-1}=2g-2\).

\subsection{AJ--Gorenstein triples}
\label{subsec:aj-gorenstein}

The block correspondence only needs two structural hypotheses, which we
isolate from the curve-specific applications.

\begin{definition}\label{def:aj-gorenstein}
The triple \((\curve,Q,X)\) is \emph{AJ--Gorenstein} if
    \begin{enumerate}[label=\textup{(\roman*)}]
        \item\textup{(Abel--Jacobi)} there exists \(\calF\in\FF_q(\curve)^*\) such that \(\divv(\calF)=D_X-nQ\);%
            \footnote{That is, \(D_X\) lies in the fibre over the origin of the Abel--Jacobi map \(\operatorname{Sym}^n\curve\to\operatorname{Pic}^0(\curve)\), \(D\mapsto[D-nQ]\).}
                \item\textup{(Gorenstein)} the semigroup \(H(Q)\) is symmetric.%
            \footnote{Such semigroups are exactly those whose semigroup ring is Gorenstein \cite{Kun70}.}
    \end{enumerate}
\end{definition}
This is standard shorthand for the divisor criterion underlying
isometry-dual one-point flags~\cite{GeilMunueraRuanoTorres2011,BrasAmorosDuursmaHong2020}.
Condition~(i) says that \(D_X\sim nQ\), with \(n\in H(Q)\); condition~(ii) gives \(K\sim(2g-2)Q\); together, they yield \(NQ-D_X\sim K\), where \(N:=n+2g-2\).  This single integer \(N\) delimits the full zero diagram; the injective zero diagram of Section~\ref{sec:injective-diagram} lives on \([0,n)\), and the twisted duality of Section~\ref{sec:full-diagram} reflects it onto all of \([0,N]\).

\begin{lemma}\label{lem:window-formula}
Let \((\curve,Q,X)\) be AJ--Gorenstein.  For \(0\le h\le N\):
\begin{enumerate}[label=\textup{(\roman*)}]
\item if \(h<n\), then \(\ev_X:\calL(hQ)\to\FF_q^n\) is injective;
\item
\(
\dim C_X(hQ)=\#\bigl(H(Q)\cap(h-n,h]\bigr).
\)
\end{enumerate}
\end{lemma}

\begin{proof}
If \(f\in\calL(hQ)\) vanishes on~\(X\), then
\(f/\calF\in\calL((h-n)Q)\).  Conversely,
\(\calF\calL((h-n)Q)\) is contained in the evaluation kernel.  Thus
\(
\ker(\ev_X|_{\calL(hQ)})
=\calF\calL((h-n)Q).
\)
When \(h<n\), the space on the right is zero.  In general,
\( 
\dim C_X(hQ)=\ell(hQ)-\ell((h-n)Q),
\)
and the semigroup formula for \(\ell\) gives the stated window count.
\end{proof}

\section{The Coweight Calculus}
\label{sec:coweight-calculus}

Since \(M_s(C)=n-d_s(C)\) (Lemma~\ref{lem:zero-weight-bridge}), the
results of this section are a coweight reformulation of Wei duality
for an arbitrary code, with no reference to $j$, to $\lambda_j$, or to
the curve.

\begin{definition}
    For a vector \(c\in\FF_q^n\), its \emph{cosupport} is the index set of its zero entries,
    \[\cosu(c)\coloneqq \{i:c_{i}=0\}.\]
    The size of this set is the \emph{coweight} of \(c\), denoted
    \[\cowt(c)\coloneqq \lvert\cosu(c)\rvert.\]
\end{definition}

\begin{definition}\label{def:max-zero-number-code}
    Let $C\subseteq\FF_q^n$ be a code; its \emph{cosupport} is the intersection of the cosupports of its codewords,
    \[
        \cosu(C)\coloneqq 
        \bigcap_{c\in C}\cosu(c)=
        \{i : c_{i}=0 \text{ for every } c\in C\}.
    \]
    If \(C\) is of dimension \(k\), its \emph{\(s\)-th generalized coweight} (\(1\le s\le k\)) is given by
    \[
        M_{s}(C) \coloneqq 
        \max\{\lvert\cosu(D)\rvert : D \subseteq C, \ \dim D=s\},
    \]
    the largest number of coordinates on which some $s$-dimensional
    subcode of $C$ vanishes.  We write $\mathcal M(C)\coloneqq
    \{M_s(C):1\le s\le k\}$ for the \emph{coweight profile} of $C$.
\end{definition}

\begin{lemma}\label{lem:zero-weight-bridge}
Let $C\subseteq\FF_q^n$ be a code of dimension $k$.
\begin{enumerate}[label=\textup{(\roman*)},leftmargin=2.4em]
\item For every subcode $D\subseteq C$ one has $\lvert\supp(D)\rvert=n-\lvert\cosu(D)\rvert$.
\item $d_{s}(C)=n-M_{s}(C)$ for $1\le s\le k$.
\end{enumerate}
\end{lemma}
\begin{proof}
\textup{(i)} A coordinate $i\in\{1,\dots,n\}$ lies in $\supp(D)$ precisely when $c_{i}\ne0$ for some $c\in D$, which occurs if and only if $i\notin\cosu(D)$; the two sets are complementary in $\{1,\dots,n\}$.

\textup{(ii)} By~\textup{(i)}, minimizing $\lvert\supp(D)\rvert$ over the $s$-dimensional subspaces $D\subseteq C$ is equivalent to maximizing $\lvert\cosu(D)\rvert$ over the same family; hence $d_s(C)=n-M_s(C)$.
\end{proof}

\begin{corollary}
\label{cor:strict-monotonicity}
Let $C\subseteq\FF_q^n$ be a code of dimension $k\ge1$. Then
$M_1(C)>\cdots>M_k(C)$; equivalently,
\[
M_s(C)\ge M_{s+1}(C)+1,
\qquad 1\le s<k.
\]
\end{corollary}
\begin{proof}
By Lemma~\ref{lem:zero-weight-bridge}\textup{(ii)}, the claim is equivalent to $d_1(C)<\dots<d_k(C)$. Let $2\le s\le k$ and let $D\subseteq C$ have $\dim D=s$ and $\lvert\supp(D)\rvert=d_s(C)$. Pick $i\in\supp(D)$ and put $D'\coloneqq\{c\in D:c_i=0\}$. The functional $c\mapsto c_i$ is nonzero on $D$, so $\dim D'=s-1$, while $\supp(D')\subseteq\supp(D)\setminus\{i\}$. Hence $d_{s-1}(C)\le\lvert\supp(D')\rvert\le\lvert\supp(D)\rvert-1=d_s(C)-1$, giving $d_{s-1}(C)<d_s(C)$ for every $2\le s\le k$, which is the claim.
\end{proof}
The strictly decreasing sequence \(M_{1}(C)>\cdots>M_{k}(C)\) constitutes the \emph{coweight hierarchy} of \(C\).

For $0\le e\le n$, put
\[
f_C(e):=\max_{\lvert S\rvert=e}\dim\{c\in C:c\vert_S=0\}.
\]

\begin{lemma}
\label{lem:coweight-threshold}
Let $C\subseteq\FF_q^n$ have dimension $k$. For $0\le e\le n$,
\[
f_C(e)=\#\{s\in\{1,\dots,k\}:M_s(C)\ge e\},
\]
and for $0\le e<n$,
\[
f_C(e)-f_C(e+1)=
\begin{cases}
1,& e\in\mathcal M(C),\\
0,& e\notin\mathcal M(C).
\end{cases}
\]
\end{lemma}
\begin{proof}
Fix $e$. If $M_s(C)\ge e$ for some $s$, choose $D\subseteq C$ of
dimension $s$ with $\lvert\cosu(D)\rvert\ge e$, and any
$S\subseteq\cosu(D)$ with $\lvert S\rvert=e$; then $D\vert_S=0$, so
$f_C(e)\ge s$. Conversely, let $S^*$ of size $e$ attain $f_C(e)$ and
put $D_0:=\{c\in C:c\vert_{S^*}=0\}$, of dimension $s_0:=f_C(e)$;
since every element of $D_0$ vanishes on $S^*$, we have $S^*\subseteq
\cosu(D_0)$, so $M_{s_0}(C)\ge\lvert\cosu(D_0)\rvert\ge\lvert
S^*\rvert=e$. Together, $f_C(e)=\max\{s:M_s(C)\ge e\}$, and since the
values $M_1(C),\dots,M_k(C)$ are pairwise distinct
(Corollary~\ref{cor:strict-monotonicity}) this maximum equals the
cardinality $\#\{s:M_s(C)\ge e\}$.

For the second claim, $\{s:M_s(C)\ge e+1\}\subseteq\{s:M_s(C)\ge e\}$,
and the difference of these two sets is $\{s:M_s(C)=e\}$, which has at
most one element by strict monotonicity — exactly one iff
$e\in\mathcal M(C)$. Hence
$f_C(e)-f_C(e+1)=\#\{s:M_s(C)\ge e\}-\#\{s:M_s(C)\ge e+1\}
=\#\{s:M_s(C)=e\}$, which is $1$ or $0$ according to whether
$e\in\mathcal M(C)$.
\end{proof}

\begin{theorem}
\label{thm:zero-profile-calculus}
Let $C\subseteq\FF_q^n$ have dimension $k$, and let $\mathbf v\in(\FF_q^*)^n$.
\begin{enumerate}[label=\textup{(\roman*)},leftmargin=2.4em]
\item For every $S\subseteq\{1,\dots,n\}$,
\[
\dim\{c\in C:c\vert_S=0\}
-\dim\{c\in C^{\perp_{\mathbf v}}:c\vert_{\{1,\dots,n\}\setminus S}=0\}
=k-\lvert S\rvert,
\]
and consequently, for $0\le e\le n$,
\[
f_C(e)-f_{C^{\perp_{\mathbf v}}}(n-e)=k-e.
\]
\item
\[
\mathcal M(C^{\perp_{\mathbf v}})
=[0,n-1]\setminus\{n-1-M_s(C): 1\leq s \leq k\}.
\]
\end{enumerate}
\end{theorem}

\begin{proof}
\textup{(i)}
Fix $S\subseteq\{1,\dots,n\}$, let $\pi_S:\FF_q^n\to\FF_q^S$ be the projection to $S$, so
\begin{equation}\label{eq:dim-zero-profile1}
k=\dim C
=\dim\pi_{S}\vert_{C}(C)+\dim\ker\pi_{S}\vert_{C}
=\dim\pi_{S}\vert_{C}(C)+\dim\{c\in C:c|_S=0\}.
\end{equation}
We claim
\[ \{d\vert_S : d\in C^\perp,\ d\vert_{\{1,\dots,n\}\setminus S}=0\}=\pi_S(C)^{\perp}\subseteq\FF_q^S. \]
Indeed, a vector $d\in\FF_q^n$ supported on $S$ verifies $d\in C^\perp$ if and only if $\sum_{i=1}^{n}d_ic_i=\sum_{i\in S}d_ic_i=0$ for every $c\in C$, which holds if and only if $d\vert_S\in\pi_S(C)^\perp$. Since restriction to $S$ is injective on vectors supported on $S$,
\[ \dim\{d\in C^\perp: d\vert_{\{1,\dots,n\}\setminus S}=0\} =\dim\pi_S(C)^\perp=\lvert S\rvert-\dim\pi_{S}\vert_{C}(C). \]

As every entry of $\mathbf v$ is a unit, the map $d\mapsto \mathbf v^{-1}\star d$ is an isomorphism $C^\perp\to C^{\perp_{\mathbf v}}$ that preserves supports; hence it restricts to an isomorphism between $\{d\in C^\perp: d\vert_{\{1,\dots,n\}\setminus S}=0\}$ and $\{c\in C^{\perp_{\mathbf v}}: c\vert_{\{1,\dots,n\}\setminus S}=0\}$, yielding
\begin{equation}\label{eq:dim-zero-profile2}
\dim\{c\in C^{\perp_{\mathbf v}}:c|_{\{1,\dots,n\}\setminus S}=0\}=\lvert S\rvert-\dim\pi_{S}\vert_{C}(C).
\end{equation}

Subtracting \eqref{eq:dim-zero-profile2} from \eqref{eq:dim-zero-profile1} eliminates $\dim\pi_{S}\vert_{C}(C)$ and gives
\[
\dim\{c\in C:c\vert_S=0\}-\dim\{c\in C^{\perp_{\mathbf v}}:c\vert_{\{1,\dots,n\}\setminus S}=0\}=k-\lvert S\rvert.
\]
For a fixed $e$, as $S$ ranges over all subsets of size $e$, the complement $\{1,\dots,n\}\setminus S$ ranges bijectively over all subsets of size $n-e$. Taking maxima on both sides gives the identity $f_C(e)-f_{C^{\perp_{\mathbf v}}}(n-e)=k-e$.

\textup{(ii)}
Fix $0\le e<n$. Applying part~\textup{(i)} with position value $n-e-1$ and with position value $n-e$ gives, respectively,
\[
f_C(n-e-1)-f_{C^{\perp_{\mathbf v}}}(e+1)=k-(n-e-1),
\qquad
f_C(n-e)-f_{C^{\perp_{\mathbf v}}}(e)=k-(n-e).
\]
Subtracting the first identity from the second,
\[
\bigl[f_C(n-e)-f_C(n-e-1)\bigr]
-\bigl[f_{C^{\perp_{\mathbf v}}}(e)-f_{C^{\perp_{\mathbf v}}}(e+1)\bigr]
=-1,
\]
that is,
\[
f_{C^{\perp_{\mathbf v}}}(e)-f_{C^{\perp_{\mathbf v}}}(e+1)
=\bigl[f_C(n-e)-f_C(n-e-1)\bigr]+1.
\]
By Lemma~\ref{lem:coweight-threshold} applied to $C$ (at position
$n-e-1$), $f_C(n-e-1)-f_C(n-e)=1$ if $(n-e-1)\in\mathcal M(C)$ and $0$
otherwise, so $f_C(n-e)-f_C(n-e-1)=-1$ if $(n-e-1)\in\mathcal M(C)$
and $0$ otherwise. Hence
\[
f_{C^{\perp_{\mathbf v}}}(e)-f_{C^{\perp_{\mathbf v}}}(e+1)
=
\begin{cases}
0,& (n-e-1)\in\mathcal M(C),\\
1,& (n-e-1)\notin\mathcal M(C).
\end{cases}
\]
By Lemma~\ref{lem:coweight-threshold} applied to
$C^{\perp_{\mathbf v}}$, this difference equals $1$ exactly when
$e\in\mathcal M(C^{\perp_{\mathbf v}})$. Therefore, for every $0\le
e\le n-1$, $e\in\mathcal M(C^{\perp_{\mathbf v}})$ if and only if
$(n-e-1)\notin\mathcal M(C)$, that is,
\[
\mathcal M(C^{\perp_{\mathbf v}})
=[0,n-1]\setminus\{n-1-m:m\in\mathcal M(C)\}
=[0,n-1]\setminus\{n-1-M_s(C):1\le s\le k\}. \qedhere
\]
\end{proof}

\begin{theorem}
\label{thm:zero-profile-extremal}
Let \(C\subseteq B\subseteq\FF_q^n\), where
\(\dim C=k\) and \(\dim B=k+1\).
\begin{enumerate}[label=\textup{(\roman*)},leftmargin=2.4em]
\item For every $1 \le s \le k$, the following interlacing inequalities hold
\[M_{s+1}(B)\le M_s(C)\le M_{s}(B).\]
\item If \(C\) contains the all-one word, then
\(
M_k(C)=0.
\)
\item If, in addition, \(k\ge2\) and the coordinate functionals
\[
\ev_i:C\to\FF_q,
\qquad
\ev_i(c)=c_i,
\]
are linearly independent for all pairs of distinct
\(i,i'\in\{1,\dots,n\}\), then
\(
M_{k-1}(C)=1.
\)
\end{enumerate}
\end{theorem}

\begin{proof}
\textup{(i)} Since \(C\subseteq B\), every $s$-dimensional subspace of
\(C\) is also an $s$-dimensional subspace of \(B\) with the same
cosupport, so \(M_s(C)\le M_s(B)\). For the other bound, let \(V\subseteq
B\) attain \(M_{s+1}(B)\), so \(\dim V=s+1\) and
\(\lvert\cosu(V)\rvert=M_{s+1}(B)\). Since \(V+C\subseteq B\) and
\(\dim B=k+1\),
\[
\dim(V\cap C)=\dim V+\dim C-\dim(V+C)\ge(s+1)+k-(k+1)=s.
\]
Choose an $s$-dimensional subspace \(D\subseteq V\cap C\). Since
\(D\subseteq V\), \(\cosu(V)\subseteq\cosu(D)\), and hence
\[
M_s(C)\ge\lvert\cosu(D)\rvert\ge\lvert\cosu(V)\rvert=M_{s+1}(B).
\]

\textup{(ii)} Suppose that \(\mathbf 1\in C\). Since \(\dim C=k\), the
only $k$-dimensional subspace of \(C\) is \(C\) itself, so
\(M_k(C)=\lvert\cosu(C)\rvert\). If \(i\in\cosu(C)\), then in
particular \(\mathbf 1_i=0\), contradicting \(\mathbf 1_i=1\). Hence
\(\cosu(C)=\varnothing\), so \(M_k(C)=0\).

\textup{(iii)}
\emph{$M_{k-1}(C)\ge1$.} For an index $i\in\{1,\dots,n\}$, \(\ev_i\) is
nonzero by hypothesis, so \(S_i:=\ker(\ev_i)\) has \(\dim S_i=k-1\) by
rank--nullity. Since \(S_i\vert_{\{i\}}=0\), we have \(i\in\cosu(S_i)\),
so \(M_{k-1}(C)\ge\lvert\cosu(S_i)\rvert\ge1\).

\emph{$M_{k-1}(C)\le1$.} Suppose that some \(S\subseteq C\) with
\(\dim S=k-1\) has \(\lvert\cosu(S)\rvert\ge2\), and pick distinct
\(i,i'\in\cosu(S)\), so \(S\subseteq\ker(\ev_i)\cap\ker(\ev_{i'})\).
Consider
\[
T:C\longrightarrow\FF_q^2,
\qquad
T(c):=(c_i,c_{i'}).
\]
If \(T\) were not surjective, its image would be a proper subspace of
\(\FF_q^2\), so there would exist \((\alpha,\beta)\ne(0,0)\) with
\(\alpha c_i+\beta c_{i'}=0\) for every \(c\in C\), i.e.\
\(\alpha\,\ev_i+\beta\,\ev_{i'}=0\), contradicting the hypothesis of
linear independence. Hence \(T\) is surjective, so
\(\dim\ker T=k-2\) by rank--nullity. But
\(\ker T=\ker(\ev_i)\cap\ker(\ev_{i'})\supseteq S\), which forces
\(k-1=\dim S\le\dim\ker T=k-2\), a contradiction. Thus every
$(k-1)$-dimensional subspace of \(C\) has cosupport of size at most
one, and consequently \(M_{k-1}(C)\le1\). Combining the two bounds,
\(M_{k-1}(C)=1\).
\end{proof}

\section{The Injective Zero Diagram (\texorpdfstring{$\lambda_j<n$}{lambda\_j<n})}\label{sec:injective-diagram}

\subsection{Coweights and zero rows}\label{subsec:base-loci}

\begin{definition}\label{def:maximum-zero-numbers}
Fix a nongap \(\lambda_j\in H(Q)\) and put \(C_j:=C_X(\lambda_jQ)\).
For \(1\le s\le k_{j}\coloneqq\dim C_j\), the \emph{$s$-th generalized
coweight of row $j$} is
\[
M_{j,s}\coloneqq M_s(C_j),
\]
as in Definition~\ref{def:max-zero-number-code}. The resulting tuple \(\mathcal M_j:=(M_{j,1},\ldots,M_{j, k_j})\) is
the \emph{\(j\)-th zero row}.
\end{definition}

Since \(C_j\subseteq\FF_q^n\) is a code for every nongap \(\lambda_j\), regardless of whether \(\ev_X\) is injective on \(\calL(\lambda_jQ)\), Lemma~\ref{lem:zero-weight-bridge}(ii), applied to \(C_j\), gives
immediately
\begin{equation}\label{eq:lower-zero-identity}
d_s(C_j)=n-M_{j,s},
\qquad 1\le s\le k_j,
\end{equation}
for every nongap \(\lambda_j\); this identity will be used again, for
\(\lambda_j\ge n\), in Section~\ref{sec:full-diagram}.

\begin{definition}\label{def:zero-diagrams}
The zero diagrams are defined as follows:
\begin{itemize}
\item The \emph{injective zero diagram} is the array
\[
\bigl(M_{j,s}\bigr)_{\lambda_j<n,\ 1\le s\le j+1}.
\]
\item   The
\emph{full zero diagram} is its extension
\( 
\bigl(M_{j,s}\bigr)_{\lambda_j\le N,\ 1\le s\le k_j}, \
k_j:=\dim C_j,
\)
in which \(M_{j,s}\) retains its meaning as a generalized coweight of
\(C_j\) but no longer, in general, as a maximum vanishing number (see
the Remark following Corollary~\ref{cor:row-boundary}).  For
\(\lambda_j<n\) one has \(k_j=j+1\); entries with \(\lambda_j<n\) are
computed directly in this section, while entries with \(\lambda_j\ge
n\), where \(k_j<j+1\) in general, are determined by reflection in
Section~\ref{sec:full-diagram}, with no new geometric input.
\end{itemize}
\end{definition}

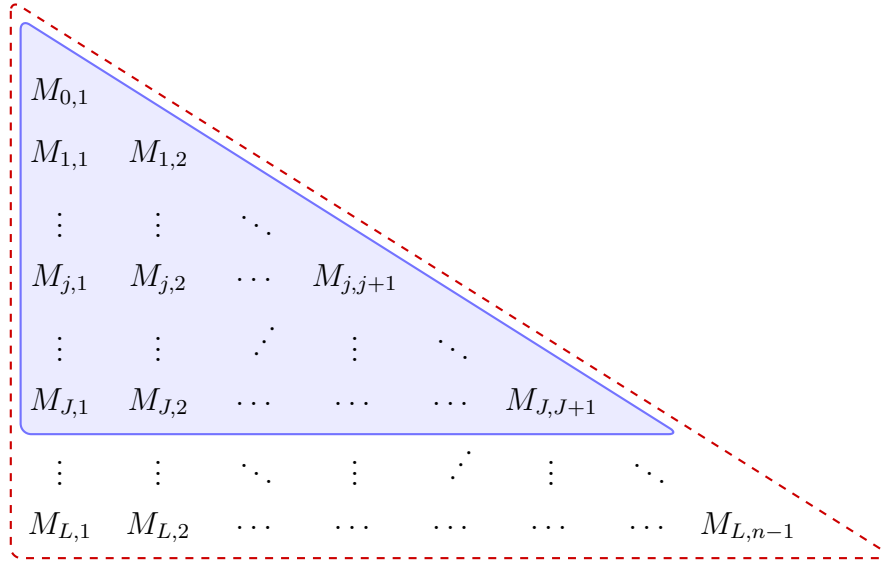
\begin{figure}[ht]
  \centering
  \begin{tikzpicture}[xscale=1.3, yscale=0.82]
    \draw[rounded corners, draw=red!80!black, thick, dashed]
      (-0.5,1.5) -- (-0.5,-7.5) -- (8.5,-7.5) -- cycle;
    \draw[rounded corners, fill=blue!8, draw=blue!55, thick]
      (-0.4,1.2) -- (-0.4,-5.5) -- (6.3,-5.5) -- cycle;
    \node at (0,0) {$M_{0,1}$};

    \node at (0,-1) {$M_{1,1}$};
    \node at (1,-1) {$M_{1,2}$};

    \node at (0,-2) {$\vdots$};
    \node at (1,-2) {$\vdots$};
    \node at (2,-2) {$\ddots$};

    \node at (0,-3) {$M_{j,1}$};
    \node at (1,-3) {$M_{j,2}$};
    \node at (2,-3) {$\dots$};
    \node at (3,-3) {$M_{j,j+1}$};

    \node at (0,-4) {$\vdots$};
    \node at (1,-4) {$\vdots$};
    \node at (2,-4) {\rotatebox{90}{$\ddots$}};
    \node at (3,-4) {$\vdots$};
    \node at (4,-4) {$\ddots$};

    \node at (0,-5) {$M_{J,1}$};
    \node at (1,-5) {$M_{J,2}$};
    \node at (2,-5) {$\dots$};
    \node at (3,-5) {$\dots$};
    \node at (4,-5) {$\dots$};
    \node at (5,-5) {$M_{J,J+1}$};

    \node at (0,-6) {$\vdots$};
    \node at (1,-6) {$\vdots$};
    \node at (2,-6) {$\ddots$};
    \node at (3,-6) {$\vdots$};
    \node at (4,-6) {\rotatebox{90}{$\ddots$}};
    \node at (5,-6) {$\vdots$};
    \node at (6,-6) {$\ddots$};

    \node at (0,-7) {$M_{L,1}$};
    \node at (1,-7) {$M_{L,2}$};
    \node at (2,-7) {$\dots$};
    \node at (3,-7) {$\dots$};
    \node at (4,-7) {$\dots$};
    \node at (5,-7) {$\dots$};
    \node at (6,-7) {$\dots$};
    \node at (7,-7) {$M_{L, n-1}$};
    
    \node[anchor=west] at (-0.5,-8) {Here, $J\coloneqq\max\{j:\lambda_{j}<n\}$ and $L\coloneqq\min\{j:\lambda_{j}\geq N\}$.};
  \end{tikzpicture}
  \caption{The zero diagram rows are indexed by all nongaps $\lambda_{j} \ge 0$, split into two regimes (Definition~\ref{def:zero-diagrams}). Within the {\color{blue}\emph{injective zero diagram} (solid outline, $\lambda_{j}<n$)}, $\mathcal{M}_{j}$ yields the generalized Hamming weights $d_{r}=n-M_{j,r}$ of the lower block. Twisted duality reflects this onto the {\color{red}\emph{full zero diagram} (dashed outline, $\lambda_{j}\leq N$)} to determine the missing weights of the upper block. The diagram naturally truncates here, as beyond it, the code is trivially the whole space. The two blocks are pole-order intervals, not array rows (see Figure~\ref{fig:blocks}). Gaps may be retained as repeated marks on the pole-order axis, but the diagram's columns are indexed by nongaps only, so that the underlying code flag is strict.}
\label{fig:zero_diagram}
\end{figure}

Within the injective zero diagram, i.e.\ for \(\lambda_j<n\), \(M_{j,s}\)
admits a second, geometric description that we use throughout the
rest of this section, and by which we shall generally refer to it
from here on.  For a subspace \(S\subseteq\calL(hQ)\) of any
Riemann--Roch space, define its \emph{vanishing set} on \(X\) by
\[
\van(S):=\{P\in X:f(P)=0\text{ for every }f\in S\},
\]
under the identification of \(X\) with \(\{1,\ldots,n\}\) fixed by
the enumeration \(X=\{P_1,\ldots,P_n\}\).  By
Lemma~\ref{lem:window-formula}(i), \(\ev_X\) is injective on
\(\calL(\lambda_jQ)\) whenever \(\lambda_j<n\), so \(S\mapsto\ev_X(S)\)
is a dimension-preserving bijection between subspaces of
\(\calL(\lambda_jQ)\) and subspaces of \(C_j\), carrying vanishing sets
to cosupports, \(\van(S)=\cosu(\ev_X(S))\); this bijection gives
\[
M_s(C_j)=\max_{\substack{S\subseteq\calL(\lambda_jQ)\\\dim S=s}}
\lvert\van(S)\rvert,
\qquad 1\le s\le\dim C_j.
\]
Moreover \(\dim C_j=\ell(\lambda_jQ)=\lvert H(Q)\cap[0,\lambda_j]
\rvert=j+1\) (the first equality by injectivity, the second by the
classical identity of Section~\ref{subsec:one-point-codes}, the third
since \(\lambda_j\) is the \((j+1)\)-st nongap).  Hence, for
\(\lambda_j<n\), \(M_{j,s}\) is exactly the \emph{maximum vanishing
number}
\[
M_{j,s}=
\max_{\substack{S\subseteq\calL(\lambda_jQ)\\\dim S=s}}
|\van(S)|,
\qquad 1\le s\le j+1,
\]
and \(\mathcal M_j=(M_{j,1},\ldots,M_{j,j+1})\), the range used for
the rest of this section.  We reserve \emph{cosupport} and
\emph{coweight} for the code-theoretic objects of
Section~\ref{sec:coweight-calculus} (vectors, codes, and their
subcodes) and \emph{vanishing set} and \emph{maximum vanishing
number} for these function-space objects \(S\subseteq\calL(hQ)\); the
two notions agree, via the bijection above, precisely on the
injective zero diagram.

\begin{remark}
The identity \(M_s(C_j)=\max_{\dim S=s}\lvert\van(S)\rvert\) uses
injectivity essentially and can fail once \(\lambda_j\ge n\): a
subspace \(S\subseteq\calL(\lambda_jQ)\) meeting \(\ker(\ev_X)\)
nontrivially can have \(\lvert\van(S)\rvert\) as large as
\(M_r(C_j)\) for some \(r<\dim S\), the kernel elements of \(S\)
vanish on all of \(X\) automatically and so impose no constraint on
\(\van(S)\), and by strict monotonicity
(Corollary~\ref{cor:strict-monotonicity}) this exceeds
\(M_{\dim S}(C_j)\).  Outside the injective zero diagram, \(M_{j,s}\)
therefore keeps its meaning as a generalized coweight of \(C_j\), but
no longer as a maximum vanishing number.
\end{remark}

\begin{definition}\label{def:separates}
A map \(\psi:X\to\FF_q^t\) \emph{separates} \(X\) if it is injective,
that is, \(\psi(P)\ne\psi(P')\) for all distinct \(P,P'\in X\).
\end{definition}

\begin{corollary}\label{cor:row-boundary}
Let \(\lambda_j<n\).
\begin{enumerate}[label=\textup{(\roman*)}]
\item \(M_{j,j+1}=0\), for every \(j\).
\item If \(j\ge1\) and \(\phi_j\) separates \(X\), then \(M_{j,j}=1\).
\end{enumerate}
\end{corollary}

\begin{proof}
Since \(\ev_X(1)\) is the all-one word, \(C_j\) contains it; apply
Theorem~\ref{thm:zero-profile-extremal}\textup{(ii)} with \(k=j+1\) to
get \(M_{j,j+1}=0\).  Suppose \(\phi_j\) separates \(X\), and let
\(i\ne i'\).  If \(\alpha\,\ev_i+\beta\,\ev_{i'}=0\) on \(C_j\), then
evaluating at the constant function \(1\in\calL(\lambda_jQ)\) gives
\(\alpha+\beta=0\); evaluating at each basis function \(f_l\) then
gives \(\alpha\bigl(f_l(P_i)-f_l(P_{i'})\bigr)=0\) for every \(l\),
and since \(\phi_j(P_i)\ne\phi_j(P_{i'})\) this forces \(\alpha=0\)
and hence \(\beta=0\).  Thus \(\ev_i\) and \(\ev_{i'}\) are
independent for all distinct \(i,i'\), so the hypothesis of
Theorem~\ref{thm:zero-profile-extremal}\textup{(iii)} holds; apply it
with \(k=j+1\) to get \(M_{j,(j+1)-1}=M_{j,j}=1\).
\end{proof}

By Corollary~\ref{cor:row-boundary}, \(M_{j,j+1}=0\); since the entries
\(M_{j,1}>M_{j,2}>\cdots>M_{j,j+1}\) are strictly decreasing nonnegative
integers by Corollary~\ref{cor:strict-monotonicity} (applied to $C_j$),
\begin{equation}\label{eq:row-minimum}
M_{j,s}\ge j+1-s,
\qquad 1\le s\le j+1.
\end{equation}

The maximum vanishing numbers form the triangular array
\[
\mathcal{M}(\curve,Q,X):=
\bigl(M_{j,s}\bigr)_{\lambda_j<n,\ 1\le s\le j+1},
\]
the \emph{injective zero diagram} of Definition~\ref{def:zero-diagrams}.
By \eqref{eq:lower-zero-identity}, computing the maximum
vanishing numbers is therefore equivalent to determining the complete
weight hierarchy of every code \(C_j\) in the flag.

Fix a basis \(\{1,f_1,\ldots,f_j\}\) of \(\calL(\lambda_jQ)\) and let
\[
\phi_j:X\longrightarrow\FF_q^j,
\qquad
P\longmapsto(f_1(P),\ldots,f_j(P))
\]
be the associated \emph{coordinate map}.  If
\(S\subseteq\calL(\lambda_jQ)\) contains no nonzero constant, its
vanishing set \(\van(S)\) is the preimage under \(\phi_j\) of an
affine subspace of \(\FF_q^j\) of codimension \(\dim S\).  A different
choice of basis for \(\calL(\lambda_jQ)\) gives an affinely equivalent
map \(\phi_j\), so \(\van(S)\) does not depend on this choice.

\subsection{Cartesian profiles}
\label{subsec:cartesian-blocks}

We begin with the affine-linear case.  This is the part of the Cartesian
theory needed below; higher-degree Cartesian evaluation codes are treated
in~\cite{BeelenDatta2018}.

\begin{theorem}
\label{thm:cartesian-affine-block}
Let \((\curve,Q,X)\) be AJ--Gorenstein and suppose \(\lambda_t<n\).  Assume
that the coordinate map
\(
\phi_t:X\longrightarrow\FF_q^t
\)
is surjective and has constant fiber size~\(\rho\), so that
\(n=\rho q^t\).  Then
\[
M_{t,s}=\rho q^{t-s}\quad(1\le s\le t).
\]
\end{theorem}

\begin{proof}
Write \(\calL(\lambda_tQ)=\langle1,f_1,\ldots,f_t\rangle\), so
\(\phi_t\) identifies \(\calL(\lambda_tQ)\) with the affine functions
on \(\FF_q^t\).  Let \(S\subseteq\calL(\lambda_tQ)\) have dimension
\(s\).  If \(S\) contains a nonzero constant, \(\van(S)=\varnothing\).

Otherwise, the linear parts of a basis of \(S\) are independent, so
they cut out \(s\) independent affine equations on \(\FF_q^t\).  Such a
system is always consistent and has exactly \(q^{t-s}\) solutions.
Hence
\[
|\van(S)|=\rho q^{t-s},
\]
attained, for instance, by \(S=\langle f_1,\ldots,f_s\rangle\), and this is the
maximum since the other case gives \(0\).  This proves the formula for
\(1\le s\le t\); together with \(M_{t,t+1}=0\)
(Corollary~\ref{cor:row-boundary}), this determines row \(t\) of the
injective zero diagram completely.
\end{proof}

\subsection{Gonality bounds}\label{subsec:gonality-bounds}

For \(r\ge0\), the rational gonality sequence is
\[
\gamma_r(\curve):=\min\{\deg A:A\ge0\text{ is $\FF_q$-rational and }
\ell(A)\ge r+1\}.
\]
Thus \(\gamma_0(\curve)=0\), and \(\gamma_1(\curve)\) is the
\(\FF_q\)-gonality \(\gamma(\curve)\).  We index by rank, following the
Brill--Noether convention also used for generalized Hamming weight
bounds in~\cite{BallicoMarcolla2016}.  Some sources instead index this
sequence by dimension, so that, for instance,
\cite{MunueraTorres2009} denotes the gonality by \(\gamma_2\) rather
than \(\gamma_1\), and the reader consulting those references should
shift indices by one accordingly.

\begin{theorem}
\label{thm:higher-rational-gonality}
Let \((\curve,Q,X)\) be AJ--Gorenstein.  If \(\lambda_j<n\), then
\[
M_{j,s}\le\lambda_j-\gamma_{s-1}(\curve),
\qquad 1\le s\le j+1.
\]
\end{theorem}

\begin{proof}
Let \(S\subseteq\calL(\lambda_jQ)\) have dimension~\(s\), and put
\(Z:=\van(S)\subseteq X\).  Form the divisor
\(
D:=\sum_{P\in Z}P,
\)
so that \(\deg D=|Z|\).  Every element of \(S\) vanishes at every point
of \(Z\), i.e.\ vanishes on \(D\), so
\[
\ell(\lambda_jQ-D)\ge s.
\]
Choose a nonzero element of this Riemann--Roch space and replace
\(\lambda_jQ-D\) by a linearly equivalent effective rational
divisor~\(\Delta\).  Then
\[
\deg \Delta=\lambda_j-|Z|,
\qquad \ell(\Delta)\ge s.
\]
Hence \(\deg \Delta\ge\gamma_{s-1}(\curve)\), which gives
\(|Z|\le\lambda_j-\gamma_{s-1}(\curve)\), that is,
\(|\van(S)|\le\lambda_j-\gamma_{s-1}(\curve)\).  Maximizing over \(S\)
gives \(M_{j,s}\le\lambda_j-\gamma_{s-1}(\curve)\).
\end{proof}

For \(s=2\), this gives the classical bound
\(
M_{j,2}\le\lambda_j-\gamma(\curve);
\)
see~\cite{Munuera1994,YangKumarStichtenoth1994}, where the same bound
is stated in terms of \(\gamma_2\) under their dimension-indexed
convention.  Its usefulness depends on knowing the rational gonality
sequence, which is a curve-specific problem.

\subsection{Full supports and product-gonality saturation}
\label{subsec:product-gonality-saturation}

\begin{definition}
\label{def:full-support-system}
For \(U\subseteq X\), put \(D_U:=\sum_{P\in U}P\), and set
\[
\mathcal G_X:=\{U\subseteq X:D_U\sim|U|Q\},
\qquad
\Phi_X:=\{|U|:U\in\mathcal G_X\}.
\]
A function \(f\) with \(\divv(f)=D_U-|U|Q\) is a \emph{full
extremizer} for \(U\).  A nongap \(\lambda_j<n\) is \emph{full} if
\(\lambda_j\in\Phi_X\).
\end{definition}

\begin{lemma}
\label{lem:full-extremizer-basics}
Let \(U\subseteq X\) be nonempty, put \(\alpha:=|U|\), and suppose
\(f\ne0\) satisfies \(\divv(f)=D_U-\alpha Q\).  Then
\(f\in\calL(\alpha Q)\setminus\calL((\alpha-1)Q)\), the zero set of
\(f\) on \(X\) is exactly \(U\), and \(\alpha\in H(Q)\).
\end{lemma}
\begin{proof}
Since \(D_U\ge0\), \(\divv(f)+\alpha Q=D_U\ge0\), so
\(f\in\calL(\alpha Q)\).  Since \(\deg\divv(f)=0\) and \(\deg D_U=\alpha\),
the pole of \(f\) at \(Q\) (its only pole) has order exactly \(\alpha\);
in particular \(f\notin\calL((\alpha-1)Q)\), so
\(\ell(\alpha Q)>\ell((\alpha-1)Q)\), i.e.\ \(\alpha\in H(Q)\).  Since
\(D_U\) is reduced with support \(U\) and \(f\) has no zeros away from
\(Q\), the zero set of \(f\) on \(X\) is exactly \(U\).
\end{proof}

\begin{lemma}
\label{lem:full-levels-are-nongaps}
Let \((\curve,Q,X)\) be AJ--Gorenstein.  Then \(\Phi_X\subseteq H(Q)\).
\end{lemma}

\begin{proof}
Let \(U\in\mathcal G_X\) and \(\alpha:=|U|\), and choose \(f\ne0\)
with \(\divv(f)=D_U-\alpha Q\).  Lemma~\ref{lem:full-extremizer-basics}
gives \(\alpha\in H(Q)\) directly.
\end{proof}

\begin{lemma}
\label{lem:full-criterion}
Let \((\curve,Q,X)\) be AJ--Gorenstein and let \(\lambda_j<n\).  Then
\[
\lambda_j\in\Phi_X \quad\Longleftrightarrow\quad M_{j,1}=\lambda_j \quad\Longleftrightarrow\quad   d_1(C_j) = n-\lambda_j
\]
\end{lemma}

\begin{proof}
\textup{($\Rightarrow$)} Suppose \(\lambda_j\in\Phi_X\), so there
exists \(U\subseteq X\) with \(|U|=\lambda_j\) and \(D_U\sim\lambda_jQ\).
Choose \(f\) with \(\divv(f)=D_U-\lambda_jQ\); by
Lemma~\ref{lem:full-extremizer-basics}, \(f\in\calL(\lambda_jQ)\) and
its zero set on \(X\) is exactly \(U\).  Taking \(S:=\langle f\rangle\), a
one-dimensional subspace with \(\van(S)=U\), gives
\(M_{j,1}\ge|\van(S)|=\lambda_j\). Conversely, any nonzero
\(g\in\calL(\lambda_jQ)\) has at most \(\lambda_j\) zeros on \(X\), since
\(\deg\divv(g)=0\) and its only allowed pole is at \(Q\) with order at
most \(\lambda_j\); hence \(M_{j,1}\le\lambda_j\). Together,
\(M_{j,1}=\lambda_j\).

\textup{($\Leftarrow$)} Suppose \(M_{j,1}=\lambda_j\). Then some
one-dimensional \(S=\langle f\rangle\subseteq\calL(\lambda_jQ)\) satisfies
\(|\van(f)|=\lambda_j\); put \(U:=\van(f)\), so \(|U|=\lambda_j\). The
divisor of zeros of \(f\) on \(X\) is at least \(D_U\), of degree
\(\lambda_j\), while its only pole lies at \(Q\) with order at most
\(\lambda_j\). Since \(\deg\divv(f)=0\), both bounds must be attained
exactly: the pole at \(Q\) has order exactly \(\lambda_j\), and \(f\)
has no zeros outside \(U\) nor multiplicities beyond \(D_U\). Hence
\(\divv(f)=D_U-\lambda_jQ\), so \(D_U\sim\lambda_jQ\), giving
\(U\in\mathcal G_X\) and \(\lambda_j=|U|\in\Phi_X\).
\end{proof}

For subspaces \(S,T\subseteq\FF_q(\curve)\), write
\[
\langle ST\rangle
:=\operatorname{Span}_{\FF_q}\{fg:f\in S,\ g\in T\}.
\]
This is the function-space counterpart of componentwise code
multiplication; see~\cite{Randriambololona2015}.

\begin{theorem}
\label{thm:product-gonality-saturation}
Let \(S\subseteq\calL(aQ)\) and \(T\subseteq\calL(bQ)\) be nonzero
subspaces of dimensions \(s\) and \(t\), and suppose
\(a+b=\lambda_J<n\).  Then
\[
\dim\langle ST\rangle\ge s+t-1,
\qquad
\van(\langle ST\rangle)=\van(S)\cup \van(T),
\]
and consequently
\begin{equation}\label{eq:product-gonality-lower}
M_{J,s+t-1}\ge |\van(S)\cup \van(T)|.
\end{equation}
If, moreover,
\begin{equation}\label{eq:product-gonality-equality}
\gamma_{s+t-2}(\curve)
=a+b-|\van(S)\cup \van(T)|,
\end{equation}
then equality holds in~\eqref{eq:product-gonality-lower}.  In
particular, if \(\lambda_j\) is full and \(\lambda_J=\lambda_j+\lambda_k\),
taking \(S\) to be a full extremizer's span and \(T=\calL(\lambda_kQ)\)
gives
\begin{equation}\label{eq:full-level-consequence}
M_{J,k+1}\ge\lambda_j,
\qquad\text{with equality when }\gamma_k(\curve)=\lambda_k.
\end{equation}
\end{theorem}

\begin{proof}
Choose bases \(f_1,\ldots,f_s\) of \(S\) and \(g_1,\ldots,g_t\) of
\(T\), ordered by their pairwise distinct pole orders at~\(Q\).  Each
product \(f_ig_l\) then has pole order
\(\ord_Q(f_i)+\ord_Q(g_l)\); among the \(s+t-1\) products
\[
f_1g_1,\ldots,f_sg_1,f_sg_2,\ldots,f_sg_t,
\]
these pole orders are pairwise distinct by construction, so the
products are linearly independent.  This proves
\(\dim\langle ST\rangle\ge s+t-1\).

For the vanishing set, a point \(P\in X\) lies in \(\van(\langle
ST\rangle)\) exactly when every product \(fg\) (\(f\in S\), \(g\in T\))
vanishes at \(P\).  Since \(\FF_q\) has no zero divisors, \(f(P)g(P)=0\)
for all such pairs if and only if either \(f(P)=0\) for all \(f\in S\),
or \(g(P)=0\) for all \(g\in T\), that is,
\(P\in \van(S)\cup \van(T)\).  This proves the base-locus identity.

The two identities together give an \((s+t-1)\)-dimensional subspace of
\(\calL(\lambda_JQ)\) with vanishing set \(\van(S)\cup \van(T)\), which
proves~\eqref{eq:product-gonality-lower}.  Under
\eqref{eq:product-gonality-equality}, the  Theorem~\ref{thm:higher-rational-gonality} gives the reverse
inequality, proving equality.  The displayed consequence follows by
taking \(S\) to be the one-dimensional span of a full extremizer of
\(\lambda_j\) (so \(s=1\)) and \(T=\calL(\lambda_kQ)\) (so \(t=k+1\)),
using \(\van(S)\cup \van(T)\supseteq \van(S)\), which has size \(\lambda_j\).
\end{proof}

\begin{theorem}
\label{thm:full-support-closure}
Let \((\curve,Q,X)\) be AJ--Gorenstein, and let \(U\in\mathcal G_X\)
have \(\lambda_j=|U|<n\), which is a nongap by
Lemma~\ref{lem:full-levels-are-nongaps}.  By
Lemma~\ref{lem:full-criterion}, a full extremizer for \(U\) evaluates
to a word of weight \(n-\lambda_j\) in \(C_X(\lambda_jQ)\), attaining
the minimum-distance bound \(d(C_X(\lambda_jQ))\ge n-\lambda_j\).
Moreover:
\begin{enumerate}[label=\textup{(\roman*)},leftmargin=2.4em]
\item \(U\in\mathcal G_X\) implies
\(X\setminus U\in\mathcal G_X\);
\item if \(U,V\in\mathcal G_X\) and \(U\cap V=\varnothing\), then
\(U\sqcup V\in\mathcal G_X\).
\item If \(U_1,\ldots,U_r\in\mathcal G_X\) are nonempty and pairwise
disjoint, and \(a_\nu=|U_\nu|\), then every subset sum and its
complement belong to \(\Phi_X\):
\[
\Sigma(a_1,\ldots,a_r)
:=\left\{\sum_{\nu\in I}a_\nu:I\subseteq\{1,\ldots,r\}\right\}
\subseteq\Phi_X,
\qquad n-\Sigma(a_1,\ldots,a_r)\subseteq\Phi_X.
\]

\end{enumerate}
\end{theorem}

\begin{proof}
Choose \(\calF\) with \(\divv(\calF)=D_X-nQ\).  If \(f_U\) is a full
extremizer for \(U\), then
\[
\divv\!\left(\frac{\calF}{f_U}\right)
=D_{X\setminus U}-(n-|U|)Q,
\]
so \(X\setminus U\in\mathcal G_X\), proving~\textup{(i)}.  If
\(U\cap V=\varnothing\), then
\(\divv(f_Uf_V)=D_{U\sqcup V}-(|U|+|V|)Q\), so \(U\sqcup V\in
\mathcal G_X\), proving~\textup{(ii)}.  Iterating~\textup{(ii)} over
all subsets \(I\subseteq\{1,\ldots,r\}\), and then applying
\textup{(i)} to each resulting full support, gives~\textup{(iii)}.
\end{proof}

\subsection{Rank and flag propagation}
\label{subsec:rank-flag-propagation}

Throughout this subsection we fix, once and for all, a nested sequence
of bases: for each nongap \(\lambda_j<n\), a basis
\(\{1,f_1,\ldots,f_j\}\) of \(\calL(\lambda_jQ)\) extending the basis chosen
for \(\calL(\lambda_{j-1}Q)\).  Such a choice is always possible, since
\(\calL(\lambda_{j-1}Q)\subsetneq\calL(\lambda_jQ)\) has codimension one
for every AJ-Gorenstein triple.
The associated coordinate map is
\[
\phi_j:X\longrightarrow\FF_q^j,
\qquad
P\longmapsto(f_1(P),\ldots,f_j(P)).
\]
The zero row records maximal intersections of \(\phi_j(X)\) with
affine linear systems.

For \(d\ge0\) and \(j\ge d+1\), set
\[
A_j^{(d)}:=M_{j,j-d}.
\]
These are the entries on the \(d\)-th diagonal of the zero diagram.

\begin{theorem}
\label{thm:rank-flag-propagation}
For every \(d\ge0\):
\begin{enumerate}[label=\textup{(\roman*)},leftmargin=2.4em]
\item if \(\lambda_{j+1}<n\), then
\[
d+1\le A_{j+1}^{(d)}\le A_j^{(d)}.
\]
Consequently, for fixed~\(d\), the sequence \((A_j^{(d)})_j\) is
nonincreasing and eventually constant.
\item if \(M_{J,J-d}=d+1\) for some \(J\ge d+1\), then
\(M_{j,j-d}=d+1\) for every \(j\ge J\) with \(\lambda_j<n\).
\item Let \(S\subseteq X\), and suppose that the affine span of
\(\phi_{J+r}(S)\), under the nested coordinates fixed above, has
dimension~\(d\) for \(0\le r\le L\).  If \(M_{J,J-d}=|S|\), then
\[
M_{J+r,J+r-d}=|S|,
\qquad 0\le r\le L.
\]
\end{enumerate}
\end{theorem}

\begin{proof}
\textup{(i)} Apply the interlacing inequality of
Theorem~\ref{thm:zero-profile-extremal}\textup{(i)} to
\(C_j\subsetneq C_{j+1}\) (of dimensions \(j+1\) and \(j+2\)) with
\(s=j-d\):
\[
A_{j+1}^{(d)}=M_{j-d+1}(C_{j+1})\;\le\;M_{j-d}(C_j)=A_j^{(d)}.
\]
By the row bound \(M_{j,s}\ge j+1-s\), applied to row \(j+1\) at
\(s=(j+1)-d\), \(A_{j+1}^{(d)}=M_{j+1,(j+1)-d}\ge d+1\). Thus
\((A_j^{(d)})_{j\ge d+1}\) is a nonincreasing sequence of integers
bounded below by \(d+1\), hence eventually constant. \textup{(ii)} In particular, once \(A_J^{(d)}=d+1\) for some
\(J\ge d+1\), monotonicity together with the floor \(d+1\) forces
\(A_j^{(d)}=d+1\) for every \(j\ge J\), that is,
\(M_{j,j-d}=d+1\).

\textup{(iii)} Fix \(0\le r\le L\). By hypothesis \(\phi_{J+r}(S)\)
lies in an affine subspace of \(\FF_q^{J+r}\) of codimension
\((J+r)-d\); this subspace is
\(\van(T)\) for some \(T\subseteq\calL(\lambda_{J+r}Q)\) with
\(\dim T=(J+r)-d\) and \(S\subseteq\van(T)\). Hence
\[
M_{J+r,(J+r)-d}\ge|\van(T)|\ge|S|.
\]
For the reverse inequality, \textup{(i)} gives
\(A_{J+i+1}^{(d)}\le A_{J+i}^{(d)}\) for \(0\le i<r\); chaining these,
\[
M_{J+r,(J+r)-d}=A_{J+r}^{(d)}\le\cdots\le A_J^{(d)}=M_{J,J-d}=|S|.
\]\end{proof}

\section{The Full Zero Diagram}\label{sec:full-diagram}

This section extends the zero diagram of
Section~\ref{sec:injective-diagram} to \(\lambda_j\le N\) by
reflection (Definition~\ref{def:zero-diagrams}), and shows that the
resulting full zero diagram determines the generalized Hamming
weights of \emph{every} one-point code \(C_X(hQ)\) with \(0\le h\le
N\), including those with \(h\ge n\), where \(\ev_X\) is no longer
injective and \(\dim C_X(hQ)<n\), namely the GHW diagram of
Definition~\ref{def:ghw-diagram}.  The mechanism is twisted duality,
together with the abstract calculus of
Theorem~\ref{thm:zero-profile-calculus}.

\subsection{Twisted duality}
\label{subsec:twisted-duality}

\begin{proposition}\label{prop:twisted-duality}
There exists \(\mathbf v\in(\FF_q^*)^n\) such that
\[
C_X(hQ)^{\perp_{\mathbf v}}=C_X((N-h)Q),
\qquad 0\le h\le N.
\]
\end{proposition}

\begin{proof}
Choose a rational differential \(\eta\) with
\(\divv(\eta)=NQ-D_X\).  It has a simple pole at every \(P_i\), and
\(v_i:=\operatorname{res}_{P_i}(\eta)\) is nonzero.  If
\(f\in\calL(hQ)\) and \(g\in\calL((N-h)Q)\), then the only possible
poles of \(fg\eta\) away from~\(Q\) are the simple poles at the points
of~\(X\).  The residue theorem gives
\[
0=\sum_{i=1}^n\operatorname{res}_{P_i}(fg\eta)
=\sum_{i=1}^n v_i f(P_i)g(P_i).
\]
Hence
\(C_X((N-h)Q)\subseteq C_X(hQ)^{\perp_{\mathbf v}}\).
By the Lemma~\ref{lem:window-formula} and the symmetry of \(H(Q)\), the two codes have
complementary dimensions, so the inclusion is an equality.
\end{proof}

We fix this \(\mathbf v\) for the remainder of the paper.

\subsection{The upper--lower block dictionary}
\label{subsec:upper-lower-correspondence}

\begin{definition}\label{def:upper-lower-blocks}
Let \(j\ge0\).
\begin{enumerate}[label=\textup{(\roman*)}]
\item The \emph{lower block~$j$} is the interval:
\(
\lambda_j\le h<\lambda_{j+1}.
\)
\item The \emph{upper block~$j$} is the reflected interval:
\(
N-\lambda_{j+1}<h\le N-\lambda_j,
\)
or, equivalently for an integer~\(h\),
\(
N-\lambda_{j+1}+1\le h\le N-\lambda_j.
\)
\end{enumerate}
\end{definition}

Every integer in a lower block gives the same Riemann--Roch space.
Twisted duality reflects that plateau into the corresponding upper
block.  Thus the two blocks are naturally indexed by the same zero row \(\mathcal M_j\).

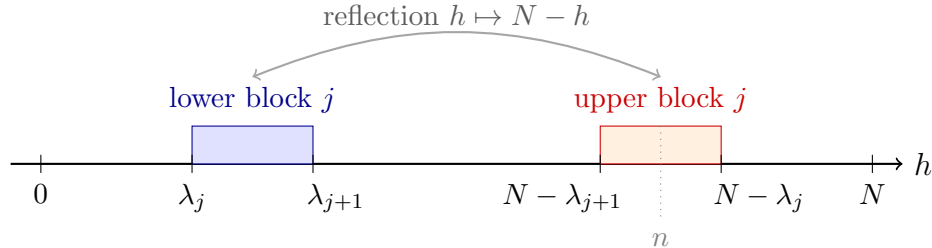
\begin{figure}[ht]
  \centering
  \begin{tikzpicture}[xscale=1.0,yscale=1.0]
    \draw[->,thick] (-0.4,0) -- (11.4,0) node[right]{$h$};

    \fill[blue!12] (2,0) rectangle (3.6,0.5);
    \draw[blue!55!black] (2,0) rectangle (3.6,0.5);
    \node[blue!55!black,font=\small] at (2.8,0.85) {lower block $j$};

    \fill[orange!12] (7.4,0) rectangle (9,0.5);
    \draw[red!80!black] (7.4,0) rectangle (9,0.5);
    \node[red!80!black,font=\small] at (8.2,0.85) {upper block $j$};

    \draw (0,0.12) -- (0,-0.12) node[below,font=\small] {$0$};
    \draw (2,0.12) -- (2,-0.12) node[below,font=\small] {$\lambda_j$};
    \draw (3.6,0.12) -- (3.6,-0.12) node[below,font=\small,xshift=3mm] {$\lambda_{j+1}$};
    \draw (7.4,0.12) -- (7.4,-0.12) node[below,font=\small,xshift=-5mm] {$N-\lambda_{j+1}$};
    \draw (9,0.12) -- (9,-0.12) node[below,font=\small,xshift=5mm] {$N-\lambda_j$};
    \draw (11,0.12) -- (11,-0.12) node[below,font=\small] {$N$};

    \draw[dotted,gray] (8.2,0.5) -- (8.2,-0.75);
    \node[gray,font=\small] at (8.2,-1.0) {$n$};

    \draw[<->,thick,gray!70] (2.8,1.15) to[bend left=22] (8.2,1.15);
    \node[gray!70!black,font=\small] at (5.5,1.95) {reflection $h\mapsto N-h$};
  \end{tikzpicture}
  \caption{The two blocks attached to row $j$ on the pole-order axis
  $[0,N]$. Every integer in the lower block gives the same
  Riemann--Roch space, hence the same code $C_j$; twisted duality
  reflects that plateau onto the upper block, which may extend beyond
  $h=n$.}
  \label{fig:blocks}
\end{figure}

\begin{theorem}
\label{thm:block-dictionary}
Let \((\curve,Q,X)\) be AJ--Gorenstein and let
\(\lambda_{j+1}\le n\).  The single zero row \(\mathcal M_j\)
determines both reflected blocks as follows.
\begin{enumerate}[label=\textup{(\roman*)},leftmargin=2.4em]
\item If \(\lambda_j\le h<\lambda_{j+1}\), then
\[
C_X(hQ)=C_j,
\qquad
d_s(C_X(hQ))=n-M_{j,s},\qquad 1\le s\le j+1.
\]
\item If \(N-\lambda_{j+1}<h\le N-\lambda_j\), then
\(\dim C_X(hQ)=n-j-1\) and
\[
\{d_r(C_X(hQ)):1\le r\le n-j-1\}
=[1,n]\setminus\{M_{j,s}+1:1\le s\le j+1\}.
\]
\end{enumerate}
\end{theorem}

\begin{proof}
The lower interval contains no nongap after \(\lambda_j\), so its
Riemann--Roch spaces equal \(\calL(\lambda_jQ)\); injectivity and the
definition of \(M_{j,s}\) prove~\textup{(i)}.  If \(h\) lies in the
upper interval, then \(N-h\) lies in the lower one.  Twisted duality
(Proposition~\ref{prop:twisted-duality}) and
Theorem~\ref{thm:zero-profile-calculus} prove
\textup{(ii)}.
\end{proof}

\subsection{Zero-row reciprocity, prefix reflection, and forced tails}
\label{subsec:complementary-rows}

The reflection of zero rows is a one-point instance of the following
Riemann--Roch identity, which does not require the Castle hypotheses
we use later.

For a divisor \(G\), put
\[
m_G(e):=\max_{\substack{U\subseteq X\\ |U|=e}}\ell(G-D_U),
\qquad D_U:=\sum_{P\in U}P.
\]

\begin{lemma}
\label{lem:divisor-zero-reciprocity}
For every divisor \(G\) and \(0\le e\le n\),
\[
m_G(e)-m_{K+D_X-G}(n-e)=\deg G+1-g-e.
\]
If \((\curve,Q,X)\) is AJ--Gorenstein, then
\[
m_{hQ}(e)-m_{(N-h)Q}(n-e)=h+1-g-e.
\]
For \(h=\lambda_j<n\), moreover,
\[
m_{\lambda_j Q}(e)
=\#\{s\in\{1,\ldots,j+1\}:M_{j,s}\ge e\}.
\]
\end{lemma}

\begin{proof}
Recall that \(K\) is the canonical divisor fixed in
Section~\ref{subsec:one-point-codes}.  By the Riemann--Roch theorem,
\begin{align*}
    \ell(G-D_{U})&=\ell(K+D_{U}-G)+\deg(G-D_{U})+1-g\\
    &=\ell(K+D_{X}-D_{X\setminus U}-G)+\deg(G-D_{U})+1-g\\
    &=\ell(K+D_{X}-G-D_{X\setminus U})+\deg(G)-e+1-g
    \intertext{Taking maxima and the bijection \(U\longmapsto X\setminus U\) between the subsets of sizes \(e\) and \(n-e\) proves}
    m_G(e)&=m_{K+D_X-G}(n-e)+\deg(G)-e+1-g,
\end{align*}
Thus obtaining the first identity
\begin{align*}
    m_G(e)-m_{K+D_X-G}(n-e)&=\deg(G)+1-g-e.
    \intertext{Assume now that \((\curve, Q, X)\) is AJ--Gorenstein. Substituting \(G=hQ\) into the formula gives}
    m_{hQ}(e)-m_{K+D_X-hQ}(n-e)&=\deg(hQ)+1-g-e,
    \intertext{Since \(\deg(hQ)=h\) and the AJ--Gorenstein property guarantees that \(NQ\sim K+D_X\)}
    m_{hQ}(e)-m_{(N-h)Q}(n-e)&=h+1-g-e.
\end{align*}
Finally, let \(h=\lambda_j<n\).  For \(U\subseteq X\), since \(D_U\)
is reduced and \(Q\notin X\),
\[
\calL(\lambda_jQ-D_U)=\{f\in\calL(\lambda_jQ):f|_U=0\},
\]
so, under the identification \(\calL(\lambda_jQ)\cong C_j\) via
\(\ev_X\), we have \(m_{\lambda_jQ}(e)=f_{C_j}(e)\) in the notation of
Lemma~\ref{lem:coweight-threshold}.  That lemma, applied to \(C_j\),
gives directly
\[
m_{\lambda_j Q}(e) = \#\{s \in \{1, \ldots, j+1\} : M_{j,s} \ge e\}.
\qedhere
\]
\end{proof}

\begin{theorem}
\label{thm:upper-prefix-reflection}
Let \((\curve,Q,X)\) be AJ--Gorenstein, let \(j\ge0\) satisfy
\(\lambda_{j+1}\le n\), and let \(h\) belong to the upper block
associated with~\(j\), that is,
\[
N-\lambda_{j+1}<h\le N-\lambda_j.
\]
Set
\[
C^+:=C_X(hQ),
\qquad
k^+:=\dim C^+=n-j-1.
\]
Suppose that, for some \(1\le r\le j+1\), the first \(r\) entries
of the lower zero row are known:
\[
M_{j,1}=A_1>A_2>\cdots>A_r=M_{j,r}.
\]
Put
\[
d_\nu:=n-A_\nu,
\qquad 1\le \nu\le r.
\]
Then
\begin{equation}
\label{eq:upper-prefix-reflection}
\mathcal M(C^+)\cap[0,d_r-2]
=
[0,d_r-2]
\setminus
\{d_1-1,\ldots,d_{r-1}-1\}.
\end{equation}
Equivalently, the integers
\[
d_1-1,\ldots,d_{r-1}-1
\]
are precisely the values missing from the coweight profile of \(C^+\)
inside the interval \([0,d_r-2]\).  Moreover,
\begin{equation}
\label{eq:upper-prefix-next-missing}
d_r-1\notin\mathcal M(C^+).
\end{equation}

The set
\[
[0,d_r-2]
\setminus
\{d_1-1,\ldots,d_{r-1}-1\}
\]
has \(d_r-r\) elements.  If \(d_r>r\), write these elements in
decreasing order as
\[
B_1>B_2>\cdots>B_{d_r-r};
\]
then the last \(d_r-r\) entries of the ordered coweight profile of
\(C^+\) are
\begin{equation}
\label{eq:upper-prefix-ordered-tail}
M_{k^+-d_r+r+\mu}(C^+)=B_\mu,
\qquad
1\le\mu\le d_r-r.
\end{equation}
When \(d_r=r\), this assertion is void.
Consequently,
\[
d_{\,k^+-d_r+r+\mu}(C^+)=n-B_\mu,
\qquad
1\le\mu\le d_r-r.
\]
\end{theorem}

\begin{proof}
Since \(h\) belongs to the upper block associated with \(j\), one has
\[
\lambda_j\le N-h<\lambda_{j+1},
\qquad\text{so}\qquad
C_X((N-h)Q)=C_j.
\]
Twisted duality (Proposition~\ref{prop:twisted-duality}) gives
\(C^+=C_j^{\perp_{\mathbf v}}\).  Applying
Theorem~\ref{thm:zero-profile-calculus} to the pair
\((C_j,C^+)\),
\begin{equation}
\label{eq:upper-complete-zero-profile}
\mathcal M(C^+)
=
[0,n-1]\setminus
\{n-1-M_{j,s}:1\le s\le j+1\}.
\end{equation}
For \(1\le\nu\le r\), the reflected value of \(M_{j,\nu}=A_\nu\) is
\(n-1-A_\nu=d_\nu-1\); since \(A_1>\cdots>A_r\), these satisfy
\(d_1-1<\cdots<d_r-1\).  In particular \(d_r-1=n-1-M_{j,r}\) is
itself one of the reflected values excluded from \(\mathcal M(C^+)\)
in~\eqref{eq:upper-complete-zero-profile}, proving
\eqref{eq:upper-prefix-next-missing}.  If \(s>r\), the strict decrease
of the lower zero row (Corollary~\ref{cor:strict-monotonicity},
applied to $C_j$) gives
\(M_{j,s}<M_{j,r}=A_r\), so
\(n-1-M_{j,s}>d_r-1\); hence no entry beyond the \(r\)-th removes a
value from \([0,d_r-2]\).  Restricting
\eqref{eq:upper-complete-zero-profile} to \([0,d_r-2]\) proves
\eqref{eq:upper-prefix-reflection}.  The right-hand side has
\((d_r-1)-(r-1)=d_r-r\) elements.  If \(d_r>r\), every remaining
entry of \(\mathcal M(C^+)\) exceeds \(d_r-2\), so these \(d_r-r\)
elements are exactly the last \(d_r-r\) entries of the ordered
coweight profile, proving \eqref{eq:upper-prefix-ordered-tail}; if
\(d_r=r\) there is nothing to prove.
\end{proof}

\begin{corollary}
\label{cor:degree-forced-upper-tail}
Let \((\curve,Q,X)\) be AJ--Gorenstein, let \(j\ge0\) satisfy
\(\lambda_{j+1}\le n\), and let
\[
N-\lambda_{j+1}<h\le N-\lambda_j.
\]
Set \(C^+:=C_X(hQ)\), \(k^+:=\dim C^+=n-j-1\), and
\(L_j:=n-\lambda_j-1\).  Then
\[
[0,n-\lambda_j-2]\subseteq\mathcal M(C^+)
\]
(empty if \(\lambda_j=n-1\)), and these are exactly the last \(L_j\)
entries of the ordered coweight profile of \(C^+\):
\[
M_{s}(C^+)=k^+-s,
\qquad
\lambda_j-j+1\le s\le k^+.
\]
Consequently
\[
d_s(C^+)=j+1+s,
\qquad
\lambda_j-j+1\le s\le n-j-1.
\]
\end{corollary}

\begin{proof}
Every nonzero function in \(\calL(\lambda_jQ)\) has at most
\(\lambda_j\) distinct zeros on \(X\), since its zero divisor has
degree at most \(\lambda_j\).  Hence \(M_{j,1}\le\lambda_j\).  Apply
Theorem~\ref{thm:upper-prefix-reflection} with \(r=1\) and
\(A_1=M_{j,1}\), so that
\[
d_1=n-M_{j,1}\ge n-\lambda_j.
\]
It follows that
\[
[0,n-\lambda_j-2]
\subseteq
[0,d_1-2]
\subseteq
\mathcal M(C^+).
\]
This interval contains \(L_j=n-\lambda_j-1\) integers, and by
Theorem~\ref{thm:upper-prefix-reflection} every remaining element of
\(\mathcal M(C^+)\) exceeds \(d_1-2\ge n-\lambda_j-2\); hence these
\(L_j\) integers are exactly the last \(L_j\) entries of the ordered
coweight profile of \(C^+\), occupying positions starting at
\[
k^+-L_j+1
=(n-j-1)-(n-\lambda_j-1)+1
=\lambda_j-j+1.
\]
Therefore \(M_{s}(C^+)=k^+-s\) for \(\lambda_j-j+1\le s\le k^+\), and
\(d_s(C^+)=n-M_{s}(C^+)=n-k^++s=j+1+s\).
\end{proof}

\section{The Postcanonical Region and the Coverage Theorem}
\label{sec:postcanonical-region}

We now specialize to the postcanonical range \(2g\le\lambda_j<n\),
where the upper block of a postcanonical row reflects \emph{back into
the injective zero diagram itself}, at the complementary row
\(j^\vee=n-j-2\), rather than out to \(h\ge n\).

\subsection{The postcanonical profile}
\label{subsec:postcanonical-profile}

A single specialization of reciprocity controls both the exact
tail and its canonical boundary.

\begin{theorem}
\label{thm:postcanonical-profile}
Let \((\curve,Q,X)\) be AJ--Gorenstein of genus \(g\ge1\), let
\(2g\le\lambda_j<n\), and set
\[
a_j:=\lambda_j-2g+2=j-g+2.
\]
Then
\begin{equation}\label{eq:postcanonical-adaptive-tail}
M_{j,j+1-\rho}=\rho,
\qquad
0\le\rho\le
\min\{j,\max\{j-g,\gamma(\curve)-2\}\}.
\end{equation}

In particular, the corresponding generalized weights satisfy
\[
d_s(C_j)=n-j-1+s
\qquad(g+1\le s\le j+1).
\]
At the canonical boundary one has the exact alternative
\begin{equation}\label{eq:canonical-boundary-fullness}
M_{j,g}=
\begin{cases}
a_j,&\text{if \(a_j\) is full},\\
a_j-1,&\text{otherwise}.
\end{cases}
\end{equation}
Here a gap is, in particular, not full.
More generally, let \(k,w\ge0\), put
\(
s:=g+k-w,
\)
and suppose that \(w<n-a_j\) and \(1\le s\le j+1\).  If there exists
\(U\subseteq X\), with \(|U|=a_j+w\), such that
\begin{equation}\label{eq:residual-configuration}
\ell(D_U-a_jQ)\ge k+1,
\qquad
w+1<\gamma_{k+1}(\curve),
\end{equation}
then
\begin{equation}\label{eq:residual-gonality-exactness}
M_{j,g+k-w}=a_j+w.
\end{equation}
If \(n>2g\), the formula
\eqref{eq:postcanonical-adaptive-tail} determines at least
\begin{equation}\label{eq:postcanonical-cell-count}
E_{\mathrm{pc}} \ge\binom{n-2g+1}{2}
\end{equation}
of the \(\binom{n-g+1}{2}\) cells in the injective zero diagram.
Writing \(\rho_{\mathrm{pc}}\) for the resulting proportion of exact
cells of the injective zero diagram, one has
\begin{equation}\label{eq:postcanonical-coverage}
\rho_{\mathrm{pc}}
\ge
\frac{(n-2g)(n-2g+1)}{(n-g)(n-g+1)}.
\end{equation}
This lower bound exceeds \(1/2\) precisely when
\begin{equation}\label{eq:half-coverage-criterion}
n>\frac{6g-1+\sqrt{8g^2+1}}{2};
\end{equation}
in particular, the simpler condition \(n>5g\) suffices.
\end{theorem}

\begin{proof}
Put \(h=\lambda_j\) and \(a=a_j\).  Since \(2g\le\lambda_j\), the index
identity \eqref{eq:postcanonical-index} gives \(h=j+g\), and hence
\(N-h=n-a\); Lemma~\ref{lem:divisor-zero-reciprocity} then gives
\begin{equation}\label{eq:postcanonical-excess-identity}
m_{hQ}(e)=j+1-e+m_{(n-a)Q}(n-e).
\end{equation}
If \(F\subseteq X\) has size \(n-e\) and \(E=X\setminus F\), then
\((n-a)Q-D_F\sim D_E-aQ\).  Hence the complementary term in
\eqref{eq:postcanonical-excess-identity} vanishes if \(e<a\), by degree,
or if \(e<\gamma(\curve)\), since otherwise a section would give a
nonconstant function with polar divisor of degree at most~\(e\).

Put \(s=j+1-\rho\).  Equation~\eqref{eq:postcanonical-excess-identity}
at \(e=\rho\) gives \(m_{hQ}(\rho)\ge s\), while the preceding vanishing
at \(e=\rho+1\) gives \(m_{hQ}(\rho+1)=s-1\) throughout the stated
range.  The conjugate-profile formula yields
\eqref{eq:postcanonical-adaptive-tail} and its two translations.

At the canonical boundary, \(e=a-1\), the same identity gives
\(m_{hQ}(a-1)\ge g\).  At \(e=a\), the complementary term has degree
zero and equals one precisely when \(n-a\), equivalently \(a\) by
Theorem~\ref{thm:full-support-closure}\textup{(i)}, is full.  Finally,
at \(e=a+1\) that term is at most one, since a degree-one divisor on a
curve of positive genus has at most one section.  Hence \(M_{j,g}\) is
exactly \(a\) in the full case and \(a-1\) otherwise.

For residual--gonality exactness, put \(s=g+k-w\) and let \(U\) satisfy
\eqref{eq:residual-configuration}.  Since \(h=a+2g-2\) and
\(K\sim(2g-2)Q\), Riemann--Roch gives
\[
\ell(hQ-D_U)
=g-1-w+\ell(D_U-aQ)
\ge g+k-w=s.
\]
Thus \(M_{j,s}\ge a+w\).  If \(U'\subseteq X\) has size \(a+w+1\),
then \(D_{U'}-aQ\) has degree \(w+1\).  Were
\(\ell(D_{U'}-aQ)\ge k+2\), a linearly equivalent effective rational
divisor of degree \(w+1\) would contradict
\(w+1<\gamma_{k+1}(\curve)\).  Hence Riemann--Roch gives
\[
\ell(hQ-D_{U'})
=g-2-w+\ell(D_{U'}-aQ)
\le g+k-w-1=s-1.
\]
No \(s\)-dimensional subspace can therefore vanish on \(a+w+1\)
evaluation points, proving~\eqref{eq:residual-gonality-exactness}.

For the cell count and coverage, assume that \(n>2g\), and put
\(R:=n-2g\).  All \(g\) gaps lie below~\(n\), so the diagram has
\(n-g\) rows and \(\binom{n-g+1}{2}\) cells.  By
\eqref{eq:postcanonical-index} its postcanonical rows are exactly
those with \(j=g+x\), where \(0\le x\le R-1\).  The degree term
\(j-g=x\) in~\eqref{eq:postcanonical-adaptive-tail} determines at
least \(x+1\) cells in that row.  Summation gives
\[
E_{\mathrm{pc}}\ge\sum_{x=0}^{R-1}(x+1)
=\binom{R+1}{2},
\]
which proves~\eqref{eq:postcanonical-cell-count} and
\eqref{eq:postcanonical-coverage}.  Finally, the inequality
\(\rho_{\mathrm{pc}}>1/2\) is equivalent to
\[
n^2-6ng+7g^2+n-3g>0.
\]
Since \(n>2g\), the relevant root is the one displayed in
\eqref{eq:half-coverage-criterion}; that root is smaller than
\((3+\sqrt2)g<5g\).
\end{proof}

\subsection{Global postcanonical staircase}
\label{subsec:global-postcanonical-staircase}

\begin{theorem}
\label{thm:global-postcanonical-staircase}
Let \((\curve,Q,X)\) be AJ--Gorenstein, let \(2g\le\lambda_j<n\), and put
\(a:=a_j=\lambda_j-2g+2\).  Suppose there exists \(b\in\Phi_X\) with
\(a\le b\le n-1\), and set \(w_0:=b-a\).  Then, for every integer \(w\)
with
\[
w_0\ \le\ w\ \le\ \min\{g-1,\ \gamma(\curve)-2\},
\qquad a+w\le n-1,
\]
one has
\begin{equation}\label{eq:global-postcanonical-staircase}
M_{j,\,g-w}=a+w,
\qquad\text{equivalently}\qquad
d_{g-w}(C_j)=n-\lambda_j+2g-2-w.
\end{equation}
\end{theorem}

\begin{proof}
Let \(F\in\mathcal G_X\) with \(|F|=b\), and choose
\(T\subseteq X\setminus F\) with \(|T|=w-w_0\); this is possible since
\(|F\sqcup T|=a+w\le n-1\le n-|F|+|F|\).  Put \(U:=F\sqcup T\), so
\(|U|=a+w\).  Since \(D_F\sim bQ\),
\[
D_U-aQ=(D_F-aQ)+D_T\sim w_0Q+D_T,
\]
which is effective, so \(\ell(D_U-aQ)\ge1\).  Apply
Theorem~\ref{thm:postcanonical-profile} with \(k=0\) and this~\(w\):
the hypotheses \(\ell(D_U-a_jQ)\ge1\) and \(w+1<\gamma_1(\curve)\)
(i.e.\ \(w\le\gamma(\curve)-2\)) of
\eqref{eq:residual-configuration} hold, as does \(1\le s=g-w\le j+1\)
in the stated range and \(w<n-a_j\) since \(a+w\le n-1\).  Equation
\eqref{eq:residual-gonality-exactness} gives
\(M_{j,g-w}=a+w\).
\end{proof}

\subsection{Postcanonical self-duality}
\label{subsec:postcanonical-self-duality}

This is the ``fold-back'' phenomenon: because
\(N-\lambda_j=n-\lambda_j+2g-2<n\) once \(\lambda_j>2g-2\), the upper
block attached to a postcanonical row does not reach a new range
above \(n\), but it is another row of the same injective zero
diagram, read in reverse.

\begin{theorem}
\label{thm:full-prefix-reflection}
Let \((\curve,Q,X)\) be AJ--Gorenstein, let \(i\ge0\) satisfy
\(\lambda_i<n\) and \(\lambda_{i+1}\le n\), put \(h:=\lambda_i\) and
\(\delta:=n-h\), and use \(\lambda_0=\gamma_0(\curve)=0\).  Suppose
that, for some \(r\ge1\),
\[
h-\lambda_k\in\Phi_X,
\qquad
\gamma_k(\curve)=\lambda_k,
\qquad 0\le k\le r.
\]
Then the first \(r+1\) entries of row~\(i\) are
\[
M_{i,k+1}=h-\lambda_k,
\qquad 0\le k\le r.
\]
\begin{enumerate}[label=\textup{(\roman*)},leftmargin=2.4em]
\item Let \(C^+:=C_X((N-h)Q)\).  Then
\[
\mathcal M(C^+)\cap[0,\delta+\lambda_r-2]
=[0,\delta+\lambda_r-2]
\setminus
\{\delta+\lambda_k-1:0\le k<r\},
\]
and the next value \(\delta+\lambda_r-1\) is also missing from
\(\mathcal M(C^+)\).
\item If, in addition, \(2g\le h<n\), then \(C^+=C_{i^\vee}\), where
\(i^\vee:=n-i-2=\delta+g-2\), and part~\textup{(i)} becomes the
identity
\[
\mathcal M_{i^\vee}\cap[0,\delta+\lambda_r-2]
=[0,\delta+\lambda_r-2]
\setminus
\{\delta+\lambda_k-1:0\le k<r\},
\qquad
\delta+\lambda_r-1\notin\mathcal M_{i^\vee}.
\]
\end{enumerate}
\end{theorem}

\begin{proof}
For each \(k\), the full level \(h-\lambda_k\) and
\(\gamma_k(\curve)=\lambda_k\) give \(M_{i,k+1}=h-\lambda_k\)
by~\eqref{eq:full-level-consequence}
(Theorem~\ref{thm:product-gonality-saturation}), applied with
\(\lambda_j=h-\lambda_k\) and \(\lambda_J=h\).

\textup{(i)} Apply Theorem~\ref{thm:upper-prefix-reflection} to
\(C^+\), with \(A_\nu=M_{i,\nu}=h-\lambda_{\nu-1}\) for
\(1\le\nu\le r+1\).  Since
\(d_\nu=n-A_\nu=n-h+\lambda_{\nu-1}=\delta+\lambda_{\nu-1}\), the
excluded values are \(\delta+\lambda_k-1\) for \(0\le k<r\).  The
missing value immediately above this interval is
\[
d_{r+1}-1=\delta+\lambda_r-1,
\]
by \eqref{eq:upper-prefix-next-missing}.

\textup{(ii)} If \(2g\le h<n\), twisted duality
(Proposition~\ref{prop:twisted-duality}) and the symmetry of \(H(Q)\)
identify \(C^+=C_X((N-h)Q)\) with the lower code \(C_{i^\vee}\), so
\(\mathcal M(C^+)=\mathcal M_{i^\vee}\), and part~\textup{(i)}
specializes to the stated identity.
\end{proof}

\begin{definition}\label{def:ghw-diagram}
The \emph{GHW diagram} is the array
\[
\bigl(d_r(C_X(hQ))\bigr)_{0\le h\le N,\ 1\le r\le k(h)},
\]
together with the missing weights at positions where \(k(h)<n\),
indexed by \emph{every integer} \(h\in[0,N]\) rather than only by
nongaps.  The missing weights are recorded as auxiliary information
about a row, not as separate cells: the coverage of
Definition~\ref{def:coverage} counts only the pairs \((h,r)\) with
\(1\le r\le k(h)\).  This is the object determined, entry by entry, by
Theorem~\ref{thm:block-dictionary}, and the object counted directly
in the coverage theorem below.
\end{definition}

\begin{definition}\label{def:coverage}
Regard the GHW diagram (Definition~\ref{def:ghw-diagram}) as a finite
set of cells, each a pair \((h,r)\) with \(0\le h\le N\) and
\(1\le r\le k(h)\).  Let \(E\subseteq\mathcal D\) be the subset of
cells whose entry is determined exactly by the results of this paper.
The \emph{coverage} of the GHW diagram is
\[
\operatorname{Cov}_{\mathrm{full}}:=\frac{|E|}{|\mathcal D|},
\]
as computed in Corollary~\ref{cor:full-diagram-coverage} below.
\end{definition}

\begin{remark}
\(\operatorname{Cov}_{\mathrm{full}}\) is coverage in the sense of
Definition~\ref{def:coverage}, measured on the GHW diagram, which has
\(n(N+1)/2\) cells indexed by every integer \(h\).  The ratio
\(\rho_{\mathrm{pc}}\) (Theorem~\ref{thm:postcanonical-profile}) is an
analogous density of exact cells within the injective zero diagram,
which has \(\binom{n-g+1}{2}\) cells indexed by nongap \(j<n\); we
reserve the term \emph{coverage} for the GHW diagram alone, since the
injective zero diagram is an auxiliary bookkeeping array, not the
object of interest.  The two cell counts differ, so the two
ratios are not directly comparable term by term.
\end{remark}

\begin{corollary}
\label{cor:full-diagram-coverage}
Let \((\curve,Q,X)\) be AJ--Gorenstein with \(n>2g\), and put
\(N=n+2g-2\).  More than half of the generalized-weight positions of
the entire code flag \(C_X(0),C_X(Q),\ldots,C_X(NQ)\) are determined
exactly: the coverage of the GHW diagram (Definition~\ref{def:coverage}) satisfies
\[
\operatorname{Cov}_{\mathrm{full}}
\ge
\frac{n(n-1)+4g}{n(n+2g-1)}
>
\frac12.
\]
\end{corollary}

\begin{proof}
For \(0\le h\le N\), write \(k(h):=\dim C_X(hQ)\), and count the cells
of the GHW diagram (Definitions~\ref{def:ghw-diagram} and~\ref{def:coverage}) as the pairs
\[
(h,r),
\qquad
0\le h\le N,
\qquad
1\le r\le k(h),
\]
that is, as the generalized-weight positions of all codes
\(C_X(hQ)\) before the flag reaches \(\FF_q^n\).

Fix \(2g\le h\le N\) and put \(t:=N-h\).  Then \(0\le t\le n-2\), so
evaluation is injective on \(\calL(tQ)\).  Every nonzero function in
\(\calL(tQ)\) has at most \(t\) zeros on \(X\); hence every entry of
the coweight profile of \(C_X(tQ)\) is at most \(t\).  By twisted
duality,
\[
C_X(hQ)=C_X(tQ)^{\perp_{\mathbf v}},
\]
and Theorem~\ref{thm:zero-profile-calculus} gives
\[
\mathcal M\bigl(C_X(hQ)\bigr)
=
[0,n-1]\setminus
\{n-1-M_{s}(C_X(tQ)):1\le s\le k(t)\}.
\]
Since \(M_{s}(C_X(tQ))\le t\), every omitted value is at least
\(n-1-t\).  Therefore
\[
[0,n-t-2]\subseteq\mathcal M\bigl(C_X(hQ)\bigr).
\]
Using \(t=N-h=n+2g-2-h\), we obtain
\[
n-t-2=h-2g.
\]
Because a coweight profile is strictly decreasing
(Corollary~\ref{cor:strict-monotonicity}), the consecutive values
\(0,1,\ldots,h-2g\) form its terminal segment: for every \(2g\le h\le
N\),
\begin{equation}
\label{eq:full-diagram-forced-tail}
M_{k(h)-\rho}\bigl(C_X(hQ)\bigr)=\rho,
\qquad
0\le \rho\le h-2g.
\end{equation}

This forced tail in row \(h\) contains \(h-2g+1\) cells.  Hence
\[
\sum_{h=2g}^{N}(h-2g+1)
=
\sum_{\rho=1}^{n-1}\rho
=
\frac{n(n-1)}{2}.
\]
For each of the remaining \(2g\) rows \(0\le h<2g\), the all-one word
gives the row-boundary identity
\[
M_{k(h)}\bigl(C_X(hQ)\bigr)=0.
\]
These rows therefore contribute \(2g\) additional exact cells, so
\begin{equation}
\label{eq:full-diagram-exact-count}
E_{\mathrm{full}}
\ge
\frac{n(n-1)}{2}+2g.
\end{equation}

To count the total number of cells, twisted duality gives
\[
k(h)+k(N-h)=n.
\]
Summing over \(0\le h\le N\) yields
\[
2\sum_{h=0}^{N}k(h)=n(N+1),
\]
so the GHW diagram contains
\begin{equation}
\label{eq:full-diagram-total-count}
D_{\mathrm{full}}
=
\sum_{h=0}^{N}k(h)
=
\frac{n(N+1)}{2}
=
\frac{n(n+2g-1)}{2}
\end{equation}
cells.  Dividing \eqref{eq:full-diagram-exact-count} by
\eqref{eq:full-diagram-total-count} gives
\[
\operatorname{Cov}_{\mathrm{full}}
\ge
\frac{n(n-1)+4g}{n(n+2g-1)}.
\]

Finally, writing \(n=2g+r\) with \(r>0\),
\[
2\bigl(n(n-1)+4g\bigr)-n(n+2g-1)
=
r(2g+r-1)+6g>0,
\]
so \(\operatorname{Cov}_{\mathrm{full}}>1/2\).
\end{proof}

\section{Specialization to Castle Curves}
\label{sec:castle-specialization}

The results of
Sections~\ref{sec:injective-diagram}, \ref{sec:full-diagram}, and
\ref{sec:postcanonical-region} hold for any AJ--Gorenstein triple.  We
now record how little extra
structure is needed to put them to work: a Castle triple supplies, at no
additional cost, a coordinate tower with uniform fibers, a universal full
packing, and an explicit two-sided gonality estimate.  Throughout this section
\((\curve,Q,X)\) denotes a Castle triple in the sense of
Definition~\ref{def:castle-curve}.

\subsection{Castle curves are AJ-Gorenstein}
\label{subsec:castle-flag}

\begin{definition}[\cite{MunueraTorres2009}]\label{def:castle-curve}
The pointed curve \((\curve,Q)\) over \(\FF_q\) is a \emph{Castle curve} if
\(H(Q)\) is symmetric and
\[
|\curve(\FF_q)|=q\lambda_1+1.
\]
For a Castle curve we take \(X=\curve(\FF_q)\setminus\{Q\}\), so
\(n=q\lambda_1\).
\end{definition}

\begin{proposition}
\label{prop:castle-is-aj-gorenstein}
Every Castle triple is AJ-Gorenstein.
\end{proposition}

\begin{proof}
Symmetry of \(H(Q)\) is condition~\textup{(ii)} of
Definition~\ref{def:aj-gorenstein} verbatim.  For condition~\textup{(i)},
let \(u\in\calL(\lambda_1Q)\setminus\calL((\lambda_1-1)Q)\); since
\(|\curve(\FF_q)|=q\lambda_1+1\) and \(u\) has a single pole of order
\(\lambda_1\) at \(Q\), every fiber of \(u:X\to\FF_q\) has size exactly
\(\lambda_1\).  Consequently \(\calF:=u^q-u\) vanishes exactly once on
each fiber, has no other zeros or poles away from \(Q\), and
\(\ord_Q(\calF)=-q\lambda_1=-n\).  Hence \(\divv(\calF)=D_X-nQ\), which
is condition~\textup{(i)}.
\end{proof}

By Lemma~\ref{lem:window-formula}, \(\dim C_X(\lambda_jQ)=j+1\) for
every nongap \(\lambda_j<n\), so the evaluation maps are injective along
the entire flag and
\[
C_0\subsetneq C_1\subsetneq\cdots\subsetneq C_j\subsetneq\cdots
\]
is strict; we refer to this as the \emph{strict Castle flag}.

\subsection{Cartesian rows and the universal fiber packing}
\label{subsec:castle-cartesian}

The single-map statement of Theorem~\ref{thm:cartesian-affine-block}
extends along an entire tower of coordinate maps once the AJ-Gorenstein
triple is Castle.

\begin{theorem}
\label{thm:cartesian-tower}
Let \((\curve,Q,X)\) be Castle, and suppose that, for some \(T\), the
map \(\phi_T:X\to\FF_q^T\) is surjective with uniform fibers.  Then for
every \(0\le j\le T\),
\[
M_{j,s}=\frac{n}{q^s}\quad(1\le s\le j).
\]
\end{theorem}

\begin{proof}
For \(1\le j\le T\), \(\phi_j\) is the projection of \(\phi_T\) onto
its first \(j\) coordinates; each fiber of \(\phi_j\) is a union of
\(q^{T-j}\) fibers of \(\phi_T\), of size \(n/q^j\).  Apply
Theorem~\ref{thm:cartesian-affine-block} with \(\rho=n/q^j\); the case
\(j=0\) is immediate.  Together with \(M_{j,j+1}=0\)
(Corollary~\ref{cor:row-boundary}), this proves the claim.
\end{proof}

We call the largest such \(T\) the \emph{Cartesian grade} of the
Castle triple.  Every Castle triple has grade at least~\(1\), since
\(\phi_1=x\) already has uniform fibers of size~\(\lambda_1\).
Uniform fibers are essential: absence of low-degree relations alone
does not imply the displayed intersection counts.

The Castle function itself, rather than a coordinate tower, also
produces a full packing of zero sets and hence exact first-column
entries far beyond the Cartesian grade.

\begin{proposition}
\label{prop:universal-fiber-packing}
Let \((\curve,Q,X)\) be Castle and choose a
Castle function \(x\in\calL(\lambda_1 Q)\).  For \(a\in\FF_q\), let
\[
F_a:=\{P\in X:x(P)=a\}.
\]
The fibers \(F_a\) are pairwise disjoint, have size~\(\lambda_ 1\), and satisfy
\(\divv(x-a)=D_{F_a}-\lambda_1Q\); hence they form a full packing of type
\((\lambda_1,\ldots,\lambda_1)\), so \(t\lambda_1\in\Phi_X\) for
\(1\le t\le q-1\).
\end{proposition}

\begin{proof}
Each \(F_a\) is a full support by construction, so by
Theorem~\ref{thm:full-support-closure}\textup{(iii)} every subset sum
\(t\lambda_1\) (\(1\le t\le q-1\)) lies in \(\Phi_X\); by
Lemma~\ref{lem:full-criterion}, \(M_{j_t,1}=t\lambda_1\).
\end{proof}

Likewise, if \(\phi_{j_0}\) separates \(X\) for some \(j_0\), then
\(M_{j_0,j_0}=1\) by Corollary~\ref{cor:row-boundary}, and
Theorem~\ref{thm:rank-flag-propagation}\textup{(ii)} (with \(d=0\))
propagates \(M_{j,j}=1\) to every \(j\ge j_0\).

\subsection{Gonality for Castle curves}
\label{subsec:castle-gonality}

For a Castle triple, the Lewittes point count converts
Theorem~\ref{thm:higher-rational-gonality} into a two-sided estimate
on the gonality itself, independent of any zero-count computation.

\begin{proposition}
\label{prop:castle-gonality}
Let \((\curve,Q)\) be Castle over \(\FF_q\).  Then
\[
\left\lceil\frac{q\lambda_1+1}{q+1}\right\rceil
\le\gamma(\curve)\le \lambda_1;
\]
in particular, \(\gamma(\curve)=\lambda_1\) whenever \(\lambda_1\le q+1\).
\end{proposition}

\begin{proof}
The Castle function \(x\in\calL(\lambda_1Q)\) already gives
\(\gamma(\curve)\le\lambda_1\).  For the lower bound, the fibres over the
\(q+1\) rational points of \(\PP^1\) under a degree-\(e\) morphism
\(\curve\to\PP^1\) contain at most \(e(q+1)\) rational points in total,
whereas
\(|\curve(\FF_q)|=q\lambda_1+1\); this forces
\(e\ge(q\lambda_1+1)/(q+1)\), and since \(e\) is an integer the ceiling
follows.  Finally, \(q\lambda_1+1>(\lambda_1-1)(q+1)\) exactly when
\(\lambda_1\le q+1\), which forces \(\gamma(\curve)>\lambda_1-1\) and
hence \(\gamma(\curve)=\lambda_1\).
\end{proof}

\begin{corollary}
\label{cor:castle-global-staircase}
Let \((\curve,Q,X)\) be Castle with Castle nongap \(\lambda_1\), let
\(2g\le\lambda_j<n\), and put \(a=a_j\).  Then
\(w_0(a):=(-a)\bmod\lambda_1\) satisfies \(a+w_0(a)\in\Phi_X\) via
the universal fiber packing
(Proposition~\ref{prop:universal-fiber-packing}), and
Theorem~\ref{thm:global-postcanonical-staircase} applies with this
\(w_0\) whenever \(w_0(a)\le\min\{g-1,\gamma(\curve)-2\}\).  In
particular every row with \(w_0(a_j)\le\min\{g-1,\gamma(\curve)-2\}\)
carries an exact staircase of length
\(\min\{g-1,\gamma(\curve)-2\}-w_0(a_j)+1\).
\end{corollary}

\begin{proof}
By Proposition~\ref{prop:universal-fiber-packing}, every multiple
\(t\lambda_1\) (\(1\le t\le q-1\)) lies in \(\Phi_X\).  Taking
\(b\) to be the least such multiple with \(b\ge a\) gives
\(w_0=b-a=(-a)\bmod\lambda_1\), and
Theorem~\ref{thm:global-postcanonical-staircase} applies verbatim.
\end{proof}

\subsection{The postcanonical density condition for Castle curves}
\label{subsec:coverage-castle}

For a Castle triple, \(n=q\lambda_1\), so the Theorem~\ref{thm:postcanonical-profile} reads entirely in terms of \(q\), \(\lambda_1\),
and \(g\).

\begin{corollary}
\label{cor:castle-coverage}
Let \((\curve,Q,X)\) be Castle, with \(q\lambda_1>2g\).  Then
\[
\rho_{\mathrm{pc}}
\ge
\frac{(q\lambda_1-2g)(q\lambda_1-2g+1)}{(q\lambda_1-g)(q\lambda_1-g+1)},
\]
and this lower bound exceeds \(1/2\) whenever
\( 
q\lambda_1>5g.
\)
\end{corollary}

\begin{proof}
Substitute \(n=q\lambda_1\) into
\eqref{eq:postcanonical-coverage} and~\eqref{eq:half-coverage-criterion}.
\end{proof}

The asymptotic form of Theorem~\ref{thm:postcanonical-profile} is 
transparent for Castle curves: it holds automatically along any family
for which the genus grows more slowly than \(q\lambda_1\).

\begin{corollary}
\label{cor:asymptotic-coverage}
Let \((\curve_i,Q_i,X_i)\) be a sequence of Castle triples with
\(n_i=q_i\lambda_1^{(i)}\to\infty\).  \[
\text{If } g_i/n_i\to0, \text{ then}\qquad
\rho_{\mathrm{pc}}(\curve_i)\longrightarrow1.
\]
\end{corollary}

\begin{proof}
Divide numerator and denominator of
\eqref{eq:postcanonical-coverage} by \(n_i^2\): both
\(\bigl(1-2g_i/n_i\bigr)\bigl(1-2g_i/n_i+1/n_i\bigr)\) and
\(\bigl(1-g_i/n_i\bigr)\bigl(1-g_i/n_i+1/n_i\bigr)\) tend to \(1\) as
\(g_i/n_i\to0\) and \(n_i\to\infty\).
\end{proof}

\section{Conclusion and Open Problems}
\label{sec:conclusions}

\subsection{Conclusions}

We introduced the zero diagram as a unified representation of
generalized coweights along an AJ--Gorenstein one-point code
flag.  Its rows encode maximal common zero sets, and the upper--lower
block correspondence shows that the same coweight data determine both
the generalized weights of short codes and the missing weights of
their reflected long-code partners in the GHW diagram of the complete
flag.  Divisor reciprocity, gonality,
full-support packings, multiplication, and rank propagation then
produce exact regions of the zero diagram from finite geometric data.

The main quantitative consequence is the GHW-diagram coverage
theorem (Corollary~\ref{cor:full-diagram-coverage}): for every
AJ--Gorenstein triple with \(n>2g\), more than half of all
generalized-weight positions in the complete flag are determined
exactly.  For the smallest Suzuki curve, combining these mechanisms determines
the large majority of the positions in the GHW diagram; see
Table~\ref{tab:suzuki-q8-theorem} for the theorem-by-theorem
breakdown.

\subsection{Open problems}
The preceding results leave three natural questions open: how tight
the coverage bounds really are, what governs the cells they miss, and
whether the GHW diagram is determined by the semigroup alone.

\begin{enumerate}[label=\textup{(\arabic*)},leftmargin=2.4em]
\item \emph{How tight is the coverage bound?}
Is \(\operatorname{Cov}_{\mathrm{full}}\ge\frac{n(n-1)+4g}{n(n+2g-1)}\)
close to optimal, or can it be improved in general, or specifically
for Castle curves?  Along a sequence of AJ--Gorenstein triples with
\(n\to\infty\), when does \(\operatorname{Cov}_{\mathrm{full}}\)
actually tend to \(1\), and when it does not, where do the uncovered
cells concentrate?

\item \emph{What determines the uncovered cells?}
The mechanisms of this paper (postcanonical exactness, duality,
Castle structure) leave some entries \(M_r(C_X(hQ))\) undetermined.
Is there a clean criterion, in terms of the gonality sequence, special
divisors, or the full-support set \(\Phi_X\), that pins these down
exactly?

\item \emph{Does the semigroup alone determine the diagram?}
If two AJ--Gorenstein triples share the same \(n\), \(g\), and
Weierstrass semigroup \(H(Q)\), must their GHW diagrams coincide, or
can the underlying geometry make them differ?  The unresolved cells of
the Suzuki example are a natural test case for this question.
\end{enumerate}

\appendix

\section{GHW-Diagram Coverage for a Suzuki Curve}
\label{app:numerical-examples}

\subsection{The GHW diagram of the smallest Suzuki curve}

\small
For $(q_0,q,g,n,N)=(2,8,14,64,90)$, the GHW diagram consists of the
$91$ codes $C_X(hQ)$, $0\le h\le90$, and contains $2912$
generalized-weight positions; the first complete code is
$C_X(91Q)=\FF_8^{64}$.

\begin{center}
\input{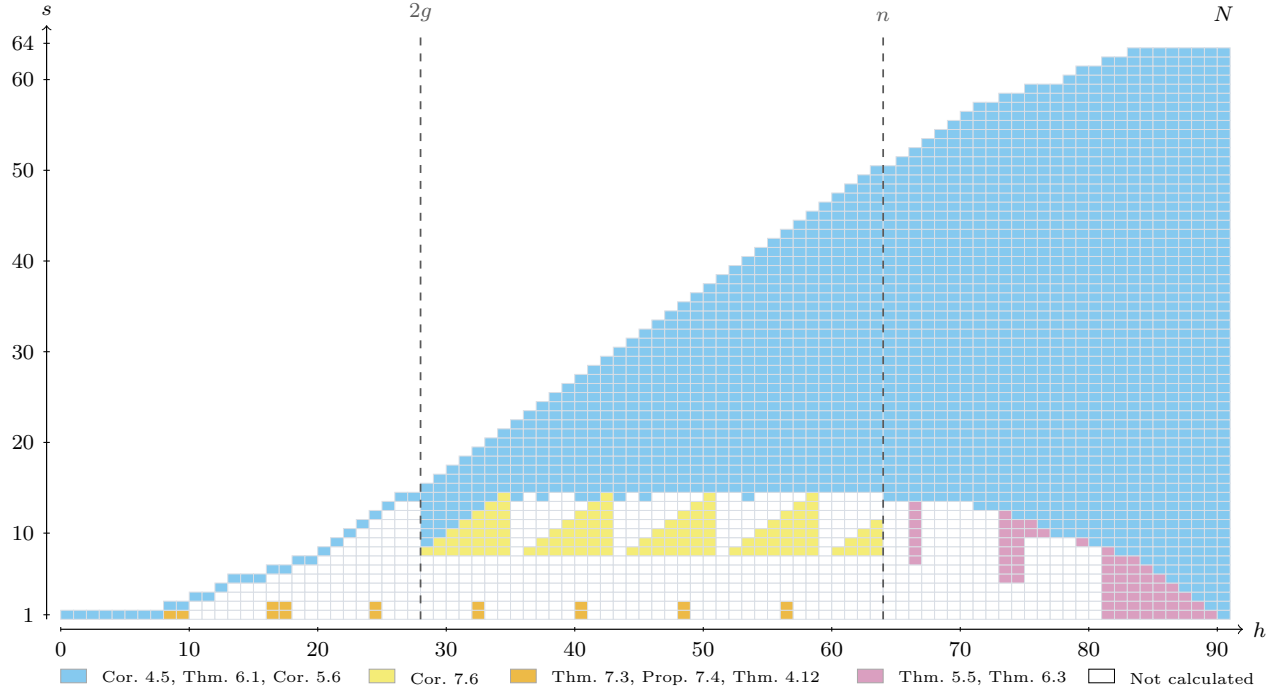}
\captionof{figure}{Coverage map of the smallest Suzuki GHW diagram.
Dashed lines mark $h=2g$ and $h=n$; colors in
Table~\ref{tab:suzuki-q8-theorem}, whose legend cites, for each
class, the theorem or corollary that supplies the corresponding
cells.}
\label{fig:suzuki-q8-full-coverage}
\end{center}

\subsection{Theorem-by-theorem coverage}

When several results determine the same cell, the cell is assigned to
the first applicable class in the descending order in Table~\ref{tab:suzuki-q8-theorem}. Thus, the contributions reported in the following table are disjoint.

\begin{center}
\small
\captionof{table}{Theorem-by-theorem breakdown for the smallest Suzuki GHW
diagram.  New column counts cells not already assigned to an earlier class.}
\label{tab:suzuki-q8-theorem}
\begin{tabular}{@{}lrrr@{}}
\toprule
Result or construction & New & Cumulative & Coverage \\
\midrule
Corollary~\ref{cor:row-boundary}, Theorem~\ref{thm:postcanonical-profile}, Corollary~\ref{cor:degree-forced-upper-tail}
   & 2{,}073 & 2{,}073 & 71.19\% \\
Corollary~\ref{cor:castle-global-staircase}
   & 122 & 2{,}195 & 75.38\% \\
Theorem~\ref{thm:cartesian-tower}, Proposition~\ref{prop:universal-fiber-packing}, Theorem~\ref{thm:product-gonality-saturation}
   & 16 & 2{,}211 & 75.93\% \\
Theorem~\ref{thm:upper-prefix-reflection}, Theorem~\ref{thm:full-prefix-reflection}
   & 69 & 2{,}280 & 78.30\% \\
\bottomrule
\end{tabular}
\end{center}

The contribution of the first result row decomposes as
\[
\begin{aligned}
2{,}073
={}&
\underbrace{28}_{0\le h<28}
+
\underbrace{\sum_{r=1}^{36}r}_{28\le h<64}
+
\underbrace{\sum_{r=37}^{63}r}_{64\le h\le90}
&+
\underbrace{21}_{\text{adaptive tail}}
+
\underbrace{7}_{\text{boundary gaps}}
+
\underbrace{1}_{\text{upper tail}}.
\end{aligned}
\]
The last three terms correspond, respectively, to the adaptive
postcanonical tail, the canonical-boundary gap cases, and one additional
cell forced by the degree-determined upper tail.


\begin{thebibliography}{99}

\bibitem{BallicoMarcolla2016}
E.~Ballico and C.~Marcolla,
\textit{Higher Hamming weights for locally recoverable codes on
algebraic curves},
Finite Fields Appl.\ \textbf{40} (2016), 61--72.

\bibitem{BarberoMunuera2000}
A.~I.~Barbero and C.~Munuera,
\textit{The weight hierarchy of {H}ermitian codes},
SIAM J.\ Discrete Math.\ \textbf{13} (2000), no.~1, 79--104.

\bibitem{BeelenDatta2018}
P.~Beelen and M.~Datta,
\textit{Generalized {H}amming weights of affine {C}artesian codes},
Finite Fields Appl.\ \textbf{51} (2018), 130--145.

\bibitem{BrasAmorosDuursmaHong2020}
M.~Bras-Amor\'os, I.~Duursma, and E.~Hong,
\textit{Isometry-dual flags of AG codes},
Des.\ Codes Cryptogr.\ \textbf{88} (2020), no.~8, 1617--1638.

\bibitem{BrasAmorosLeeVicoOton2014}
M.~Bras-Amor\'os, K.~Lee, and A.~Vico-Oton,
\textit{New lower bounds on the generalized {H}amming weights of {AG} codes},
IEEE Trans.\ Inform.\ Theory \textbf{60} (2014), no.~10, 5930--5937.
doi:10.1109/TIT.2014.2343993.

\bibitem{CampsMoreno2025}
E.~Camps-Moreno, H.~H.~L\'opez, G.~L.~Matthews, and R.~San-Jos\'e,
\textit{The weight hierarchy of decreasing norm-trace codes},
Des.\ Codes Cryptogr.\ \textbf{93} (2025), no.~7, 2873--2894.

\bibitem{DuursmaKirov2009}
I.~Duursma and R.~Kirov,
\textit{An extension of the order bound for {AG} codes},
in: Applied Algebra, Algebraic Algorithms and Error-Correcting Codes
(AAECC 2009), Lecture Notes in Comput.\ Sci.\ \textbf{5527},
Springer, Berlin, 2009, pp.~11--22.

\bibitem{DuursmaPark2008}
I.~Duursma and S.~Park,
\textit{Coset bounds for algebraic geometric codes},
Finite Fields Appl.\ \textbf{16} (2010), no.~1, 36--55.

\bibitem{FarranEtAl2018}
J.~I.~Farr\'an, P.~A.~Garc\'ia-S\'anchez, B.~A.~Heredia, and M.~J.~Leamer,
\textit{The second {F}eng--{R}ao number for codes coming from
telescopic semigroups},
Des.\ Codes Cryptogr.\ \textbf{86} (2018), no.~9, 1849--1864.

\bibitem{GeilMunueraRuanoTorres2011}
O.~Geil, C.~Munuera, D.~Ruano, and F.~Torres,
\textit{On the order bounds for one-point {AG} codes},
Adv.\ Math.\ Commun.\ \textbf{5} (2011), no.~3, 489--504.
doi:10.3934/amc.2011.5.489.

\bibitem{HansenStichtenoth1990}
J.~P.~Hansen and H.~Stichtenoth,
\textit{Group codes on certain algebraic curves with many rational
points},
Appl.\ Algebra Engrg.\ Comm.\ Comput.\ \textbf{1} (1990), no.~1,
67--77.

\bibitem{Kun70}
E.~Kunz,
\textit{The value-semigroup of a one-dimensional Gorenstein ring},
Proc. Amer. Math. Soc. \textbf{25} (1970),  no.~4, 748--751.

\bibitem{Lee2015}
K.~Lee,
\textit{Bounds for generalized {H}amming weights of general {AG} codes},
Finite Fields Appl.\ \textbf{34} (2015), 265--279.
doi:10.1016/j.ffa.2015.02.006.

\bibitem{Matthews2004}
G.~L.~Matthews,
\textit{Codes from the {S}uzuki function field},
IEEE Trans.\ Inform.\ Theory \textbf{50} (2004), no.~12, 3298--3302.

\bibitem{MontanucciTimpanellaZini2018}
M.~Montanucci, M.~Timpanella, and G.~Zini,
\textit{AG codes and AG quantum codes from cyclic extensions of the Suzuki and the Ree curves},
J.\ Geom.\ \textbf{109} (2018), no.~1, Paper~23, 18 pp.

\bibitem{Munuera1994}
C.~Munuera,
\textit{On the generalized {H}amming weights of geometric {G}oppa codes},
IEEE Trans.\ Inform.\ Theory \textbf{40} (1994), no.~6, 2092--2099.

\bibitem{MunueraTorres2009}
C.~Munuera, A.~Sep\'ulveda, and F.~Torres,
\textit{Castle curves and codes},
Adv.\ Math.\ Commun.\ \textbf{3} (2009), no.~4, 399--408.

\bibitem{OlayaLeonGranados2015}
W.~Olaya-Le\'on and C.~Granados-Pinz\'on,
\textit{The second generalized {H}amming weight of certain {C}astle codes},
Des.\ Codes Cryptogr.\ \textbf{76} (2015), no.~1, 81--87.

\bibitem{Randriambololona2015}
H.~Randriambololona,
\textit{On products and powers of linear codes under componentwise multiplication},
in: S.~Ballet, M.~Perret, and A.~Zaytsev (eds.),
\textit{Algorithmic Arithmetic, Geometry, and Coding Theory},
Contemporary Mathematics \textbf{637},
American Mathematical Society, Providence, RI, 2015, pp.~3--78.
doi:10.1090/conm/637/12749.

\bibitem{Stichtenoth2009}
H.~Stichtenoth,
\textit{Algebraic Function Fields and Codes},
2nd ed., Graduate Texts in Mathematics \textbf{254}, Springer, Berlin, 2009.

\bibitem{Wei1991}
V.~K.~Wei,
\textit{Generalized {H}amming weights for linear codes},
IEEE Trans.\ Inform.\ Theory \textbf{37} (1991), no.~5, 1412--1418.

\bibitem{YangKumarStichtenoth1994}
K.~Yang, P.~V.~Kumar, and H.~Stichtenoth,
\textit{On the weight hierarchy of geometric {G}oppa codes},
IEEE Trans.\ Inform.\ Theory \textbf{40} (1994), no.~3, 913--920.

\end{thebibliography}
\end{document}